\documentclass[a4paper,11pt]{article}

\usepackage[utf8]{inputenc}
\usepackage[T1]{fontenc}
\usepackage[english]{babel}

\usepackage{indentfirst}
\usepackage{amsmath}
\usepackage{amsthm}
\usepackage{amssymb}
\usepackage{graphicx}
\usepackage{subfigure}
\usepackage{psfrag}{}
\usepackage{color}
\usepackage{pgf}

\allowhyphens

\newcommand{\valos}{\mathbb{R}}
\newcommand{\complex}{\mathbb{C}}

\newcommand{\eps}{\varepsilon}

\newtheorem{thm}{Theorem}

\newtheorem{thmlem}{Lemma}

\newcommand{\Pe}{\mathcal{P}}

\newcommand{\GGG}{\mathcal{G}}

\newcommand{\ket}[1]{{\left|#1\right\rangle}}
\newcommand{\bra}[1]{{\left\langle #1\right|}}

\newcommand{\skalarszorzat}[2]{{\langle #1 | #2 \rangle}}

\usepackage{ifpdf}

\ifpdf
\usepackage{epstopdf}
\usepackage[pdftex,ps2pdf,dvips,colorlinks,urlcolor=blue,citecolor=blue,linkcolor=blue]{hyperref}
\else
\usepackage[hypertex,ps2pdf,dvips,colorlinks,urlcolor=blue,citecolor=blue,linkcolor=blue]{hyperref}
\fi

\makeatother

\begin{document}

\numberwithin{equation}{section}

\title{On Form Factors in nested Bethe Ansatz systems}
\author{Bal\'azs Pozsgay, Willem-Victor van Gerven Oei, M\'arton Kormos}
\date{\today}

\maketitle

\abstract{
We investigate form factors of local operators in the multi-component
Quantum Non-linear Schr\"odinger model, a prototype theory solvable by
the so-called nested Bethe Ansatz. We determine the analytic
properties of the infinite volume form factors using the
coordinate Bethe Ansatz solution and we establish a connection with the finite
volume matrix elements. In the two-component models we derive a set of
recursion relations for the ``magnonic form factors'', which are the
matrix elements on the nested Bethe Ansatz states. In certain simple
cases (involving states with only one spin-impurity) we obtain
explicit solutions for the recursion 
relations. 
}

\section{Introduction}

One of the goals of many-body quantum physics is the calculation of
correlation functions of local observables. 
 The form factor program is an approach to tackle this problem; it
 consists of the following three steps:
 \begin{enumerate}
 \item Finding the eigenstates of the system and inserting a complete
   set of states between the two (or more) local operators.
\item Evaluating the matrix elements (form factors) of the local
  operators.
\item Summing up the resulting spectral series.
 \end{enumerate}
In generic models these tasks present a fantastic challenge. However,
the situation is quite different in one-dimensional integrable models,
where
there are
exact methods available to compute the exact spectrum and the form
factors.

One class of models where this program has been particularly
successful is the Bethe Ansatz solvable theories related to the $sl(2)$ symmetric
$R$-matrix, such as the Lieb-Liniger model (also known as the Quantum
Nonlinear Schr\"odinger equation) \cite{Lieb-Liniger,Lieb2} and the
XXX and XXZ spin chains \cite{XXX,XXZ1,XXZ2,XXZ3}. A very powerful
approach is the 
Algebraic Bethe Ansatz, developed by the Leningrad-school
\cite{Faddeev:1979gh,faddeev-history,Faddeev-ABA-intro}, which led to
important results concerning the scalar products of Bethe states
\cite{korepin-norms,slavnov-overlaps} and the form factors
\cite{springerlink:10.1007/BF01029221,kojima-korepin-slavnov}.
Tremendous effort has been devoted to the calculation of correlation
functions as well; we do not attempt here to review the literature, instead
we refer the reader to the book \cite{korepinBook} and the recent
paper \cite{2011arXiv1101.1626K}.

A different class of models are those non-relativistic theories, where
the excitations over a fixed reference state have internal degrees
of freedom, for example the multi-component Non-Linear Schr\"odinger
equation, 
the $sl(N)$ symmetric spin chains,
or the one-dimensional Hubbard model. The spectrum of these models can
be obtained by the so-called nested Bethe Ansatz
\cite{nested-S-1,McGuire-BA,nested-S-2,Yang-nested,Yang-nested-S,sutherland-nested,all-spin-chains,HubbardBook},
the algebraic formulation of which was worked out for the
$sl(N)$-related models in the papers
\cite{Kulish-FZ-algebra,kulish-resh-glN}. Although the nested Bethe
Ansatz is successful in finding the spectrum, the construction of
the eigenstates is rather complicated and there are far fewer results
available than in the $sl(2)$ case. Norms of eigenstates were obtained
 in
\cite{resh-su3,pang-zhao-norms,Hubbard-norms,
tarasov-varchenko1,tarasov-varchenko2} and there are approaches to
calculate the scalar products, see \cite{pakuliak-ragoucy}
and references therein. However, no compact and convenient formulas
have been found yet, which would facilitate the computation of
correlation functions.

In the non-relativistic models mentioned above there are explicit and
exact representations known for the eigenstates of the system; this
allows (at least in principle) for constructive methods to find the form factors and
correlation functions. The situation is different in the realm of
(massive) integrable relativistic QFT's \cite{Mussardo:1992uc}. These
theories are typically investigated in infinite volume, 
the Hilbert-space is spanned by asymptotic scattering states defined
using the Faddeev-Zamolodchikov algebra \cite{zam-zam}. This
construction does not allow the direct determination of form factors
of local operators; an indirect method has been developed instead: the
so-called form factor bootstrap program
\cite{Karowski:1978vz,Berg-Karowski-Weisz,smirnov_ff,Karowski-LSZ,bootstrap-osszfogl,Smirnov-qKZ1,Smirnov-qKZ2}. The
idea is to establish the analytic properties of the form
factors based on the requirement of locality, 
resulting in a closed set of equations also called ``form factor
axioms''. These equations are restrictive enough so that
supplied with a few additional assumptions (possibly depending on the
operator in question) they uniquely determine the form factors. In
massive field theories it is sufficient to obtain explicit expressions
for the form factors with a small number of particles, because they typically
saturate the spectral series for the vacuum correlations
even at small distances \cite{zam_Lee_Yang,balog-weisz-structure}. 

One of the most important form factor axioms is the kinematic pole
(or annihilation pole)
property, which relates $(N,M)$ form factors (matrix elements on an
$N$-particle and an $M$-particle state) to $(N-1,M-1)$ form
factors\footnote{In relativistic field theory the $(N,M)$ form factors
can be expressed in terms of the analytic continuation of the
$(0,N+M)$ form factors using the so-called crossing relation. Then the
kinematic pole is usually expressed in terms of the $(0,N)$ form
factors relating them to the $(0,N-2)$ matrix elements.}. It states
that whenever particle rapidities from the two 
states approach each other the form factor has a simple pole (kinematic
singularity) and the amplitude is given by the form factor with the
corresponding particles not present and a pre-factor depending on the
exact $S$-matrix of the theory \cite{smirnov_ff,Karowski-LSZ}. 
Similar singularity properties were also found in the models related
to the $sl(2)$-symmetric $R$-matrix in the framework of the
Algebraic Bethe Ansatz. The pole structure of the scalar products  of
Bethe states was first established in \cite{korepin-norms}, this led
to the discovery of the celebrated Slavnov-formula
\cite{slavnov-overlaps}, a determinant formula describing the scalar
product of an eigenstate and an arbitrary Bethe state. These
results refer to the finite volume states and they were derived using
a quite general
algebraic construction. Moreover, they were used to determine form
factors of local operators, and in the case when both states are
eigenstates, the singularity properties of the resulting form factors are
found to be essentially the same as in the relativistic case. This has been noted recently in
\cite{nonrelFF}, where it was also shown that a special
non-relativistic and small-coupling limit of the sinh-Gordon model
form factors yields the known matrix elements of the Lieb-Liniger model.

We wish to note that form factors of local and composite non-local
operators in the Lieb-Liniger model were also considered using the
infinite volume Quantum Inverse Scattering Method
\cite{PhysRevD.21.1523,PhysRevD.23.3081}. 
One result of this approach is the so-called quantum Rosales
expansion, which expresses the local field operators using the
non-local Faddeev-Zamolodchikov operators. The Rosales expansion can
be used to read off explicit expressions for the form factors, and to
establish their analytic properties
\cite{creamer-thacker-wilkinson-rossz}, leading to the same
kinematical pole equation (apart from the normalization) as the one
found in \cite{iz-kor-resh,korepinBook,nonrelFF}. 

The kinematical pole axiom also appears in the study of the form factors
of the anti-ferromagnetic spin chain \cite{Pakuliak-XXZ-FF}. In this
case the states involved are excitations above the infinite volume anti-ferromagnetic
ground state, which is already filled with a finite
density of elementary particles. In this respect the situation
considered in \cite{Pakuliak-XXZ-FF} is different than in
\cite{korepin-norms,slavnov-overlaps,nonrelFF}, where the states
in question are elementary excitations over the reference state. 

\bigskip

In this paper we contribute to the calculation
of form factors in the multi-component Nonlinear Schr\"odinger
equation. Inspired by the results of \cite{nonrelFF} we revisit the
methods of relativistic QFT:
 we consider the analytic properties of form
factors, set up recursion relations and make an attempt
to find solutions to them, without trying to find manageable
expressions for scalar 
products or related, more basic quantities. 

In obtaining explicit solutions to the recursion relations we restrict
ourselves to matrix elements on states with only a single
spin-impurity. One of our motivations to investigate this subclass of form
factors is provided by a recent  experiment with ultracold atomic gas 
\cite{impurity-spin-exp-koehl}, where the motion of spin-impurities
was studied in an otherwise polarized background. Our present results can form
the basis for the theoretical investigation of such
situations. Related questions were studied in a very recent article
\cite{zvonarev-impurity}. The paper \cite{zvonarev-impurity}  only considered the
infinite coupling case, whereas our results for the form factors hold
at arbitrary coupling strengths. 

One of the main steps of the present work is the identification
of the (un-normalized) form factors in finite and infinite
volume. This result bears relevance also for integrable
relativistic QFT, where related questions have been investigated
recently \cite{Takacs-Feher1,Takacs-Feher2}. We give a few remarks on this issue in the Conclusions.

The structure of the paper is as follows. In section 2 we consider
the one-component case (the Lieb-Liniger model) and establish a number
of results about the form factors, using only the coordinate Bethe
Ansatz wave functions. Although this section does not
contain new results, it serves as a basis for the generalizations in
later sections.
In section 3 we recall the
construction of the (infinite volume) Bethe Ansatz states in the
multi-component case 
and we establish the properties of the form factors, in particular the
kinematical pole property. Section 4 deals with the two-component
case: the magnonic form factors are introduced, which are the matrix
elements on the nested Bethe Ansatz states. A set of ``magnonic form
factor equations'' is obtained. In section 5 we solve these equations
in a number of simple cases, involving states with a single
spin-impurity. The sections 3-5 are concerned with the infinite
volume situation, the connection to the finite volume nested Bethe
Ansatz states is made in section 6.
Finally, section 7 is devoted to our conclusions.

\section{The Lieb-Liniger model: Coordinate Bethe Ansatz}

\label{LL}

In this section we review the basic facts about the
coordinate Bethe Ansatz solution  of the one-component Bose gas, the
Lieb-Liniger model. The ideas and results 
of this section will be the basis for our investigations of the
multi-component systems in section \ref{multicomponent}, \ref{2c} and \ref{fftcsa1hehe}.

\subsection{The model and the coordinate Bethe Ansatz solution}

The second quantized form of the Hamiltonian is
\begin{equation}
\label{H-LL}
H=\int_{-L/2}^{L/2}\,\mathrm d
x\left(\partial_x\Psi^\dagger\partial_x\Psi+ 
c\Psi^\dagger\Psi^\dagger\Psi\Psi\right).
\end{equation}
Here
$\Psi(x,t)$ and $\Psi^\dagger(x,t)$ are  canonical
non-relativistic Bose
fields satisfying
\begin{equation}
  [\Psi(x,t),\Psi^\dagger(y,t)]=\delta(x-y).
\end{equation}
We used the conventions  $m=1/2$ and $\hbar=1$. 
The first quantized form of the Hamiltonian is
\begin{equation*}
H=-\sum_{j=1}^N\frac{\partial^2}{\partial x_j^2}+2c\sum_{j<l} \delta(x_j-x_l).
\end{equation*}
The parameter $c$ is the
coupling constant; in the present work we only consider the repulsive
case $c>0$.

In \eqref{H-LL} $L$ denotes the size of the system. We will consider
both the infinite volume ($L=\infty$) and 
finite volume cases. In the latter case we always assume  periodic
boundary conditions.

The eigenstates of the Hamiltonian \eqref{H-LL} can be constructed
using the Bethe Ansatz \cite{Lieb-Liniger,Lieb2,korepinBook}.
The $N$-particle coordinate space wave function is given by
\begin{equation}
\label{egyfajta-coo}
  \chi_N(\{x\}_N|\{p\}_N)=\frac{1}{\sqrt{N!}} \sum_{\Pe\in S_N} (-1)^{[\Pe]}
\exp\left\{ i\sum_j x_j (\Pe p)_j \right\} \prod_{j>k}
\Big((\Pe p)_j-(\Pe p)_k-ic\epsilon(x_j-x_k)\Big),
\end{equation}
where $\{p\}_N$ is the set of pseudo-momenta or rapidities, $\epsilon(x)$ is the
sign function, and the $\Pe\in S_N$ are permutations.
The total energy and momentum of the multi-particle state is 
\begin{equation*}
  E_N=\sum_j p_j^2\qquad\qquad P_N=\sum_j p_j.
\end{equation*}

In the infinite volume case the wave function
\eqref{egyfajta-coo} is an eigenstate for arbitrary set of
rapidities. 
Moreover, the Bethe states with real rapidities form a complete set of
states \cite{gaudin-book,opdam-completeness,tracy-widom}. Rapidities
with non-zero imaginary parts are not allowed because they result in
unbounded wave functions.

In the finite volume case periodic
boundary conditions force the quasi-momenta to be solutions 
of the Bethe Ansatz equations:
\begin{equation}
  \label{LL-BA-e}
  e^{ip_jL}\prod_{k\ne j} \frac{p_j-p_k-ic}{p_j-p_k+ic}=1,
  \qquad\qquad
j=1\dots N.
\end{equation}
It is known that in the repulsive case ($c>0$) considered here all
solutions to the Bethe equations are given by real rapidities and they provide
 a complete set of states \cite{YangYang2,dorlas-completeness}.

In finite volume the wave function \eqref{egyfajta-coo} is
normalizable. The norm of the eigenstates is given by \cite{gaudin-LL-norms,korepin-norms}
\begin{equation}
  \label{egyfajta-norm}
\mathcal{N}^{LL}(\{p\}_N)= \int  \ |\chi_N|^2=\prod_{j<k} ((p_j-p_k)^2+c^2)\times
\det \mathcal{G}^{LL}
\end{equation}
with
\begin{equation}
  \mathcal{G}^{LL}_{jk}=\delta_{j,k}\left(
L+\sum_{l=1}^N \varphi(p_j-p_l)\right) -\varphi(p_j-p_k)
\end{equation}
and
\begin{equation*}
  \varphi(u)=\frac{2c}{u^2+c^2}.
\end{equation*}

\subsection{Form Factors in finite and infinite volume}

We are interested in the form factors of the field and density
operators:
\begin{equation*}
  \Psi(0) \qquad\text{and}\qquad \rho(0)=\Psi^\dagger(0)\Psi(0).
\end{equation*}
The finite volume form factors are defined as
\begin{equation}
\label{Fabadef}
\begin{split}
&  \mathbb{F}^\Psi_N(\{p\}_N,\{k\}_{N+1})=\\
&\sqrt{N+1} \int_{-L/2}^{L/2}  dx_1 \dots dx_N\
\chi_N^*(x_1,\dots,x_N|\{p\}_N)\chi_{N+1}(0,x_1,\dots,x_N|\{k\}_{N+1})
\end{split}
\end{equation}

\begin{equation}
\label{Fabadef2}
\begin{split}
&  \mathbb{F}^\rho_N(\{p\}_N,\{k\}_{N})=\\
&N \int_{-L/2}^{L/2} dx_1 \dots dx_{N-1}\
\chi_N^*(0,x_1,\dots,x_{N-1}|\{p\}_N)\chi_{N}(0,x_1,\dots,x_{N-1}|\{k\}_{N}).
\end{split}
\end{equation}

As coordinate space integrals these form factors are well defined for
arbitrary rapidities and they 
depend on the volume through the parameters
\begin{equation*}
  l(k_j)=e^{ik_jL}\qquad\qquad l^*(p_j)=e^{-ip_jL}.
\end{equation*}
The dependence on these parameters was studied thoroughly using
Algebraic Bethe Ansatz \cite{korepin-izergin,korepin-LL1,iz-kor-resh,korepinBook}.
In particular, it was shown that if
the rapidities are solutions to the Bethe equations, and there are no
coinciding rapidities, then the form
factors do not depend on the volume explicitly (apart from possible
overall phase factors). 

Note that the form factors are defined
using the un-normalized wave functions, therefore the actual finite
volume matrix elements can be obtained as
\begin{equation}
\label{fanswer}
  \bra{\{p\}_N}\Psi\ket{\{k\}_{N+1}}=\frac{ \mathbb{F}^\Psi_N(\{p\}_N,\{k\}_{N+1})}
{\sqrt{\mathcal{N}^{LL}(\{p\}_N)\mathcal{N}^{LL}(\{k\}_{N+1})}},
\end{equation}
and similarly for the density operator.

An alternative definition for the form factors can be given in
infinite volume. In this case the real space integrals are not  convergent
due to the oscillating wave functions. However, they can be made
convergent by introducing regulators $f_\eps(x)$ in $x$-space.
We choose
\begin{equation*}
  f_\eps(x)\equiv f(\eps |x|),
\end{equation*}
where $f(x): \valos^+\to \valos^+$ is a smooth function
satisfying
\begin{equation*}
\lim_{x\to 0}   f(x)=1\qquad\qquad \lim_{x\to\infty} f(x)=0.
\end{equation*}
An example is given by $f(x)=e^{-x}$.
It can be shown that for every $p\in \valos\setminus \{0\}$
\begin{equation}
\label{rules}
\lim_{\eps\to 0} \int_{0}^\infty dx\  f_\eps(x) e^{ipx}
=\frac{i}{p}\qquad\qquad
\lim_{\eps\to 0} \int_{-\infty}^0 dx\ f_\eps(x) e^{ipx} =-\frac{i}{p}
\end{equation}
independently of the choice of $f(x)$. Actually \eqref{rules} can be
considered as a well-defined prescription to evaluate the infinite
volume integrals.

Using this prescription the infinite volume form factors are defined
as
\begin{equation}
\label{Finfdef}
\begin{split}
  \mathcal{F}^\Psi_N(\{p\}_N,\{k\}_{N+1})&=
\lim_{\eps\to 0}\sqrt{N+1} \int_{-\infty}^\infty  dx_1 \dots dx_N\
\prod_{j=1}^N f_\eps(x_j) \\
&\hspace{1.5cm}\times \chi_N^*(x_1,\dots,x_N|p)\chi_{N+1}(0,x_1,\dots,x_N|k) 
\end{split}
\end{equation}
\begin{equation}
\begin{split}
\label{Finfdef2}
  \mathcal{F}^\rho_N(\{p\}_N,\{k\}_{N})&=
\lim_{\eps\to 0}
N \int_{-\infty}^\infty dx_1 \dots dx_{N-1}\  \prod_{j=1}^{N-1} f_\eps(x_j)\\
&\hspace{1.5cm}\times  \chi_N^*(0,x_1,\dots,x_{N-1}|p)\chi_{N}(0,x_1,\dots,x_{N-1}|k).
\end{split}
\end{equation}

The connection between the finite volume and infinite volume form
factors is established by the following theorem.

\begin{thm}
\label{fftcsa}
The form factors are the same in finite and infinite volume. In other
words, if both sets $\{p\}$ and $\{k\}$ are solution to the Bethe
equations with a volume parameter $L$ and there are no coinciding rapidities ($p_j\ne k_l$), then 
\begin{equation}
\label{dubmates2a}
   \mathbb{F}^\Psi_N(\{p\}_N,\{k\}_{N+1})=
 \mathcal{F}^\Psi_N(\{p\}_N,\{k\}_{N+1})
\end{equation}
\begin{equation}
\label{dubmates2}
   \mathbb{F}^\rho_N(\{p\}_N,\{k\}_{N})=
 \mathcal{F}^\rho_N(\{p\}_N,\{k\}_{N}).
\end{equation}
\end{thm}
\begin{proof}
For simplicity we only consider the field operator and the case
$N=1$. Then we have to prove the equation 
  \begin{equation}
\label{dubmates}
    \lim_{\eps\to 0}\int_{-\infty}^\infty dx \ f_\eps(x) \
\chi_1^*(0|p)\ \chi_2(0,x|k_0,k_1)=
\int_{-L/2}^{L/2} dx \  \
\chi_1^*(0|p)\ \chi_2(0,x|k_0,k_1).
\end{equation}
The integrand consists of sums of free wave functions with certain amplitudes. Due to the
insertion of the field operator the amplitudes depend on the sign of
$x$, therefore the
integrals have to be split into
two parts:
\begin{equation*}
 \int_{-\infty}^{\infty}=\int_{-\infty}^{0}+\int_{0}^{\infty}\qquad\text{and}\qquad
  \int_{-L/2}^{L/2}=\int_{-L/2}^{0}+\int_{0}^{L/2}.
\end{equation*}
The integral of an exponential function can be evaluated in the
infinite volume case by \eqref{rules} whereas in the finite volume
case it is given by the
Newton-Leibniz formula. The essential step to prove \eqref{dubmates}
is to note that in the finite volume case those terms of the
Newton-Leibniz formula which 
represent the contributions at $x=0$ exactly coincide with the
corresponding contributions of the infinite volume case.
On the other hand, the contributions of the Newton-Leibniz formula
corresponding to $-L/2$ and $L/2$ cancel each other owing to the
periodicity of the Bethe wave function. Therefore, the two sides of
\eqref{dubmates} are indeed equal.

Similar arguments can be given in the case of higher particle form
factors, and also for the form factors of the density operator.
\end{proof}

It is clear from this derivation that the case of coinciding
rapidities, in particular the problem of expectation values  has to be
treated separately. Whenever there are coinciding rapidities the infinite 
volume FF becomes divergent, the finite volume FF remains finite and
can be expressed using the properly defined limits of the infinite
volume ones \cite{fftcsa2}. In the framework of Algebraic Bethe Ansatz such relations
were established for certain non-local operators related to
correlation functions in
\cite{korepin-izergin,korepin-LL1,iz-kor-resh}, whereas local operators
describing higher-body local correlations were considered recently  in
\cite{LM-sajat}. 
In this work we will only consider the case of non-coinciding
rapidities. 

\subsection{Important properties of the form factors}

\label{FFproperties}

The coordinate Bethe Ansatz wave functions \eqref{egyfajta-coo} are
completely anti-symmetric with respect to an exchange of two
rapidities. Therefore
\begin{equation*}
  \mathcal{F}_N^\Psi(p_1,\dots,p_N|k_0,\dots,k_j,k_{j+1},\dots,k_N)=
-\mathcal{F}_N^\Psi(p_1,\dots,p_N|k_0,\dots,k_{j+1},k_{j},\dots,k_N)
\end{equation*}
\begin{equation*}
  \mathcal{F}_N^\Psi(p_1,\dots,p_j,p_{j+1},\dots,p_N|k_0,\dots,k_N)=
-\mathcal{F}_N^\Psi(p_1,\dots,p_{j+1},p_{j},\dots,p_N|k_0,\dots,k_N),
\end{equation*}
and similarly for the density operator.

The form factors have kinematical pole singularities whenever $p_j\to
k_l$ for some $j,l$. The residue of the pole is given by Theorem
\ref{hukko} below. However, before establishing the theorem we need the
following lemma:

\begin{thmlem}
\label{grills}
  Let $D_j\subset \valos^{N-1}$, $j=0,\dots,N-1$ be the region
  \begin{equation*}
    x_1 < x_2 < \dots < x_j < 0 < x_{j+1} < \dots < x_{N-1}.
  \end{equation*}
Then integrating over this region we obtain
\begin{equation}
  \label{whitebear}
\lim_{\eps\to 0}  \int_{D_j} e^{i\sum_{l=1}^{N-1}
  p_lx_l}\prod_{l=1}^{N-1}f_\eps(x_l) =
\prod_{l=1}^{j}   \frac{-i}{\sum_{m=1}^l p_m}
\prod_{l=j+1}^{N-1} \frac{i}{\sum_{m=l}^{N-1} p_m}.
\end{equation}
\end{thmlem}
\begin{proof}
  The lemma is proven easily by induction, i.e. by
performing first the integral over
$x_1$ or $x_{N-1}$.
\end{proof}

\begin{thm}
\label{hukko}
Let $p_N\to k_N$. Then the behaviour of $\mathcal{F}^\Psi_N$  is given
by
\begin{equation}
\label{interloper}
\begin{split}
&\mathcal{F}^\Psi_N(p_1,\dots,p_N|k_0,\dots,k_N)\sim \\
&\frac{i}{k_N-p_N}\left(
\prod_{j=1}^{N-1} (p_{Nj}+ic)  \prod_{j=1}^{N} (k_{Nj}-ic)
-
\prod_{j=1}^{N-1} (p_{Nj}-ic)  \prod_{j=1}^{N} (k_{Nj}+ic)   
\right)\times \\
&\hspace{6cm}\mathcal{F}^\Psi_{N-1}(p_1,\dots,p_{N-1}|k_0,\dots,k_{N-1}).
\end{split}
\end{equation}
\end{thm}
\begin{proof}
The form factor is given by a sum over two sets of permutations
and a sum over the different regions. 
 The pole in $p_N-k_N$ only
appears for those permutations when both rapidities are coupled to the same $x_j$. Moreover
it follows from Lemma \ref{grills}
 that the singularity only appears in those regions,
where either $x_j$ is larger than all other coordinates 
(including $x=0$ for the position of the field operator), or if $x_j$
is smaller than all other coordinates. In these cases the integral over
$x_j$ yields (assuming that the second largest or second smallest
coordinate is $x_l$)
\begin{equation*}
\begin{split}
& \lim_{\eps\to 0} \int_{x_l}^\infty dx_j\ f_\eps(x_j) e^{i
  (k_N-p_N)x_j}\quad \to\quad \frac{-1}{i(k_N-p_N)} 
  e^{i(k_N-p_N)x_l}\quad \sim \quad
\frac{-1}{i(k_N-p_N)}\\
& \lim_{\eps\to 0}   \int^{x_l}_{-\infty} dx_j\  f_\eps(x_j) e^{i
  (k_N-p_N)x_j}\quad \to\quad \frac{1}{i(k_N-p_N)} 
  e^{i(k_N-p_N)x_l}\quad \sim \quad
\frac{1}{i(k_N-p_N)}.\\
\end{split}
\end{equation*}
Collecting all these terms and adding the proper pre-factors which
arise due to the rapidities $p_N$ and $k_N$ according to
\eqref{egyfajta-coo} we obtain \eqref{interloper}.
\end{proof}

There is an analogous relation for the density operator.
The form factors of both operators
 have the structure
\begin{equation}
\label{rmbreturns}
  \mathcal{F}_N(\{p\},\{k\})=
\prod_{j<l} (k_j-k_l)\prod_{j>l}
 (p_l-p_j) \prod_{j,l} \frac{1}{p_j-k_l}\times P_N(\{p\}|\{k\}),
\end{equation}
where $P_N(\{p\}|\{k\})$ are polynomials symmetric in both sets of
variables. The degrees of the polynomials in their variables are
established by the following lemmas.

\begin{thmlem}
  \label{field-asy}
 The asymptotic behaviour of the field operator form factor is
 $p_1^{N-3}$ at $p_1\to\infty$ and $k_0^{N-1}$ at $k_0\to\infty$.
\end{thmlem}
\begin{proof}
  At $p_1\to\infty$ the overall degree of the wave function is
  $p_1^{N-1}$. The leading term factorizes: the amplitude does not
  depend on the position of the particle with rapidity $p_1$.
 Therefore the integral over the coordinate attached to
  $p_1$ can be performed over the whole real line and the
  regularization scheme yields
  \begin{equation*}
\lim_{\eps\to 0}    \int_\valos f_\eps(x_1) e^{ip_1x_1}=
\lim_{\eps\to 0}  \left(  \int_{\valos^+} f_\eps(x_1)
  e^{ip_1x_1}+\int_{\valos^-} f_\eps(x_1) e^{ip_1x_1} \right)=
\frac{i}{p_1}+\frac{-i}{p_1}=
0.
  \end{equation*}
Therefore the overall degree of the form factor is determined by the sub-leading terms of
order $p_1^{N-2}$. The highest possible degree of the coordinate space
integrals is $p_1^{-1}$, resulting in an overall degree of $p_1^{N-3}$. 

At $k_0\to\infty$ the leading term in the real space integral for the
form factor is given by those terms where $k_0$ is attached to the
coordinate $x_0=0$. To leading order the reamining wave function is
proportional to a
Bethe Ansatz state with the rapidities $\{k_1,\dots,k_N\}$ and the
integrals yield the scalar product
$\skalarszorzat{p_1,\dots,p_N}{k_1,\dots,k_N}$. It is assumed that the
rapidities are different therefore this scalar product vanishes. The
sub-leading terms in the wave function then yield an overall degree of
$k_0^{N-1}$. 
\end{proof}

It follows from the Lemma above and the factorization
\eqref{rmbreturns} that
in the case of the field operator $P_N$ is of degree $N-1$
in all of the $p_j$ and the $k_j$ variables.
Therefore the recursion relations \eqref{interloper} 
contain enough
information which completely determine the form factors; the
polynomial $P_N$ can be reconstructed using 
the Lagrange interpolation procedure \cite{iz-kor-resh}.

In the case of the density operator the polynomial
$P_N(\{p\}|\{k\})$, which is symmetric with respect to the exchange
$\{p\}\leftrightarrow\{k\}$, has the following asymptotic behaviour:
\begin{thmlem}
\label{density-asy}
  The density form factor satisfies the asymptotics
  \begin{equation}
    \label{density-asymptotics}
\lim_{p_1\to\infty} \frac{F_N^\rho(p_1\dots p_N|k_1\dots k_N)}{p_1^{N-1}}=
(-1)^{N-1}F_{N-1}^\Psi(p_2\dots p_N|k_1\dots k_N).
  \end{equation}
\end{thmlem}
\begin{proof}
 The leading terms in the $p_1\to\infty$ limit are given by those
 permutations where $p_1$ is attached to the coordinate 0, in other
 words it is not integrated over. Concerning these terms the following
 relation can be read off from \eqref{egyfajta-coo}
 \begin{equation}
\label{ezalimit}
   \lim_{p_1\to\infty}
   \frac{\chi_N(0,x_1,\dots,x_{N-1}|p_1,\dots,p_N)}{p_1^{N-1}}
\sim
\frac{(-1)^{N-1}}{\sqrt{N}} \chi_{N-1}(x_1,\dots,x_{N-1}|p_2,\dots,p_N),
 \end{equation}
where the sign $\sim$ indicates that on the l.h.s. only those
permutations $\mathcal{P}\in S_N$ are kept which leave $p_1$ at the
first place. The statement of the theorem then follows directly from
\eqref{ezalimit} and the definitions \eqref{Finfdef}-\eqref{Finfdef2}.
\end{proof}

It follows that $P_N$ is of order $N$ in all of its variables. The
kinematic recursion relations together with the condition
\eqref{density-asymptotics} completely fix the form factor.

\subsection{Determinant formulas for the form factors}

For the sake of completeness we present here explicit determinant
formulas for the form factors, which 
will be the basis for
the generalizations in section \ref{solutions}.

In the case of the field operator the form factor can be expressed as
\cite{kojima-korepin-slavnov,korepin-slavnov} 
\begin{equation*}
  \mathcal{F}_N^\Psi= \prod_{j\ne l} ((k_j-k_l)^2+c^2)
\det(S_{jl}-S_{N+1,l}).
\end{equation*}
Here $S$ is an $N\times N$ matrix defined as
\begin{equation*}
  S_{jl}=t(p_j,k_l)\frac{\prod_{m=1}^N (p_m-k_j+ic)}{\prod_{m=0}^N
    (k_m-k_j+ic)}-
t(k_l,p_j)\frac{\prod_{m=1}^N (k_j-p_m+ic)}{\prod_{m=0}^N
    (k_j-k_m+ic)},
\end{equation*}
with
\begin{equation*}
  t(u)=\frac{-c}{u(u+ic)}
\end{equation*}

The first three examples are given explicitly by
\begin{equation*}
  \mathcal{F}^\Psi_0(k)=1,\
\qquad\qquad
    \mathcal{F}^\Psi_1(p|k_0,k_1)=\frac{2c(k_0-k_1)}{(k_0-p)(k_1-p)},
\end{equation*}
\begin{equation*}
\begin{split}
&   \mathcal{F}^\Psi_2(p_1,p_2|k_0,k_1,k_2)=\frac{-4c^2(k_0-k_1)(k_1-k_2)(k_0-k_2)(p_1-p_2)}
{(k_0-p_1)(k_0-p_2)(k_1-p_1)(k_1-p_2)(k_2-p_1)(k_2-p_2)}\times \\
&\hspace{1cm}\times (c^2 - k_0 k_1 - k_0 k_2 - k_1 k_2 + (k_0+k_1+k_2)(p_1+p_2)
- 3 p_1 p_2).
\end{split}
\end{equation*}

Concerning the density operator, the first determinant formula was
established in 
\cite{springerlink:10.1007/BF01029221}. In the present work we will
use an independent representation, which is easily derived from the
results of \cite{XXZ-to-LL-sajat}:
\footnote{The final form of the determinant formula \eqref{FFK} (and
the generalizations
  \eqref{FFKMSF}  and \eqref{FFKM}) was suggested by
  Jean-S\'ebastien Caux. The result given in \cite{XXZ-to-LL-sajat}
  was expressed as a sum of $N$ determinants, 
whereas the present
  formula is given by a single determinant, making it more convenient
  for numerical calculations.
}
\begin{equation}
  \label{FFK}
  \begin{split}
&    \mathcal{F}^\rho_N(\{p\},\{k\})=\frac{-i}{c} (-1)^{N(N+1)/2}
\prod_{o=1}^N \prod_{l=1}^N (k_o-p_l+ic)\times 
\det V.
  \end{split}
\end{equation}
Here $V$ is an $(N+1)\times (N+1)$ matrix with entries
\begin{equation}
  \begin{split}
 V_{jl}&= \tilde t(k_j,p_l)+ \tilde t(p_l,k_j)
\prod_{o=1}^N
\frac{(p_l-k_o+ic)(p_l-p_o-ic)}{(p_l-k_o-ic)(p_l-p_o+ic)} \qquad \qquad  j,l=1\dots N\\
V_{N+1,j}&=\prod_{o=1}^N \frac{p_o-p_j+ic}{k_o-p_j+ic} \qquad\text{and}\quad
V_{j,N+1}=1,\qquad \qquad  j=1\dots N\\
V_{N+1,N+1}&=0,
  \end{split}
\end{equation}
and
\begin{equation*}
\tilde t(u)=\frac{-i}{u(u+ic)}.
\end{equation*}
We wish to note that \eqref{FFK} can be written alternatively as an $N\times
N$ determinant, but it is useful to keep this form, which makes it possible
to find generalizations in section \ref{solutions}.

The first two cases are given explicitly as
\begin{equation*}
\begin{split}
  \mathcal{F}_1^\rho(p|k)&=1\\
 \mathcal{F}_2^\rho(p_1,p_2|k_1,k_2)&=
\frac{-2c (k_1+k_2-p_1-p_2)^2 (k_1-k_2)(p_1-p_2)}{(k_1-p_1)(k_1-p_2)(k_2-p_1)(k_2-p_2)}.
\end{split}
\end{equation*}

\subsection{An alternative representation for the form  factors}

The determinant formulas of the previous subsection 
are very convenient for both analytical and numerical analysis of the
correlation functions. However, it is possible to find alternative
representations, which might not seem as useful at first sight, but
which might give clues for the calculation of form factors in nested
Bethe Ansatz systems. 

One such representation can be derived using
results of the form factor bootstrap program of relativistic
integrable QFT's \cite{bootstrap-osszfogl}. The papers \cite{Babujian:2001xn,Babujian:2002fi} 
considered the form factors of breathers (in particular the
lowest-lying breathers) in the sine-Gordon model, and
they arrived at formulas, which give the form factors with a total
number of $N$ particles as a sum of $2^N$ terms. 
From this result it
is possible to derive formulas for 
the Lieb-Liniger model, first performing an analytic continuation in
the coupling constant to
get the form factors of the sinh-Gordon model, and then 
 using a non-relativistic and small-coupling limit as explained in
 \cite{sinhG-LL1,sinhG-LL2,nonrelFF}.  

In the case of the field operator we obtain the formula
\begin{equation}
\label{finally}
\begin{split}
   \mathcal{F}^{\Psi}_N(\{p\}_N|\{k\}_{N+1})=\frac{P}{2c^{N}}
\prod_{j,k} \frac{1}{k_j-p_l}.
\end{split}
\end{equation}
Here the polynomial $P$ is given by
\begin{equation}
\begin{split}
 P=&\mathop{\sum_{\alpha_j=0,1}}
\sum_{\beta_l=0,1}  (-1)^{\sum_j \alpha_j+\sum_l \beta_l }
\prod_{i1<i2} (p_{i1}-p_{i_2}+(\alpha_{i_1}-\alpha_{i_2})ic)\times \\
&\times \prod_{i1<i2} (k_{i1}-k_{i_2}+(\beta_{i_1}-\beta_{i_2})ic)
\prod_{i1,i2} (p_{i1}-k_{i_2}-(\alpha_{i_1}-\beta_{i_2})ic)  \times\\
&\times 
\left(  \sum_{j=1}^N (-1)^{\alpha_j}+ \sum_{j=1}^{N+1}
  (-1)^{\beta_j}\right).
\label{PP1}
\end{split}
\end{equation}
The summation is performed over the variables $\alpha_j=0,1$ with
$j=1\dots N$ and
$\beta_j=0,1$, with $j=0\dots N$. In general the number of the
$\alpha_j$ and $\beta_j$ variables coincides with the number of
particles in the bra and ket states, respectively.

In the case of the density operator the corresponding formula reads
\begin{equation}
\label{jaapeden}
\begin{split}
&   \mathcal{F}^{\rho}_N(\{p\}_N|\{k\}_{N},\mu)=\frac{-P}{8c^{N-1}}
\prod_{j,k} \frac{1}{k_j-p_l}
\end{split}
\end{equation}
with
\begin{equation}
\begin{split}
 P=&\mathop{\sum_{\alpha_j=0,1}}
\sum_{\beta_l=0,1}  (-1)^{\sum_j \alpha_j+\sum_l \beta_l }
\prod_{i1<i2} (p_{i1}-p_{i_2}+(\alpha_{i_1}-\alpha_{i_2})ic)\times \\
&\times \prod_{i1<i2} (k_{i1}-k_{i_2}+(\beta_{i_1}-\beta_{i_2})ic)
\prod_{i1,i2} (p_{i1}-k_{i_2}-(\alpha_{i_1}-\beta_{i_2})ic)  \times\\
&\times
\left( \sum_{j=1}^N \left((-1)^{\alpha_j}+ (-1)^{\beta_j}\right)\right)^2.
\end{split}
\label{PP2}
\end{equation}
Note that the only difference between formulas \eqref{PP1} and
\eqref{PP2} is the last factor, which is called the ``p-function'' in the original papers
\cite{Babujian:2001xn,Babujian:2002fi} .

It can be shown that formulas \eqref{finally}-\eqref{jaapeden}
satisfy all necessary conditions established 
in subsection \ref{FFproperties}, therefore they describe the field
operator and density form factors indeed. In the case of \eqref{jaapeden} the
factorization condition \eqref{density-asymptotics} is also easily checked 
using 
\begin{equation*}
\begin{split}
&\left(1+ \sum_{j=2}^N (-1)^{\alpha_j}+
\sum_{j=1}^N (-1)^{\beta_j} \right)^2-
\left(-1+ \sum_{j=2}^N (-1)^{\alpha_j}+
\sum_{j=1}^N (-1)^{\beta_j} \right)^2=\\
&\hspace{7cm}4\left(\sum_{j=2}^N (-1)^{\alpha_j}+
\sum_{j=1}^N (-1)^{\beta_j} \right).
\end{split}
\end{equation*}

\section{Multi-component systems: Coordinate Bethe Ansatz}

\label{multicomponent}

In this section we consider the general $K$-component models in infinite volume. In
second quantized form the Hamiltonian is 
\begin{equation}
\label{HM-LL}
H=\int_{-\infty}^{\infty}\,\mathrm d
x\left(\partial_x\Psi^\dagger_j\partial_x\Psi_j+ 
c\Psi_l^\dagger\Psi_j^\dagger\Psi_j\Psi_l\right).
\end{equation}
Here
$\Psi_j(x,t)$ and $\Psi_j^\dagger(x,t)$, $j=1\dots K$ are  canonical
non-relativistic Bose or Fermi fields satisfying
\begin{equation}
  \Psi_j(x,t)\Psi_l^\dagger(y,t)-\sigma \Psi_l^\dagger(y,t)\Psi_j(x,t)=\delta_{jl}\delta(x-y),
\end{equation}
where $\sigma=1$ for bosons ($\sigma=-1$ for fermions), respectively.

In first quantized form the Hamiltonian is
\begin{equation*}
H=-\sum_{j=1}^N\frac{\partial^2}{\partial x_j^2}+2c\sum_{j<l} \delta(x_j-x_l)
\end{equation*}
in both cases. Note that the above Hamiltonian is completely
spin-independent; the interaction between the different spin
components arises as the effect of the statistics of the wave function.

The construction of the eigenstates of the Hamiltonian \eqref{HM-LL}
was established in the papers
\cite{nested-S-1,McGuire-BA,nested-S-2,Yang-nested,Yang-nested-S} (for
a more general scheme see \cite{stokman-opdam-spin-case}). In
the following  
we collect the main results of this procedure;  our focus will be on the
form factors and their analytic properties.

\subsection{The wave functions}

Before constructing the coordinate Bethe Ansatz wave functions we need
to introduce a few basic objects and notations.

Consider the vector space $V=\complex^K$. Consider also the $N$-fold tensor product
\begin{equation*}
  V^{(N)}=\otimes^N V
\end{equation*}
and a representation $\rho$ of the permutation
group $S_N$ on $V^{(N)}$. Let $\rho_{ab}$ denote the action
corresponding to the elementary exchange $(ab)$. In the physical cases
\begin{equation*}
  \rho_{ab}=\sigma P_{ab},
\end{equation*}
where $P_{ab}$ is the permutation operator between the vector spaces
$V_a$ and $V_b$ and $\sigma=1$ ($\sigma=-1$) in the bosonic
(fermionic) case, respectively.

Consider also a set of parameters $\{p\}_N$, which will play the role
of particle rapidities for the Bethe wave function.

We introduce the operators \cite{Yang-nested,stokman-opdam-spin-case}
\begin{equation}
\label{yangsY}
  Y^{ab}_{jk}=\frac{(p_j-p_k)\rho_{ab}-ic}{p_j-p_k+ic}.
\end{equation}
Here it is understood that $Y^{ab}_{jk}$ acts only on the vector
spaces $a$ and $b$ and the indices $jk$ stand for the rapidities
entering the expression \eqref{yangsY}. 

The operators \eqref{yangsY} satisfy the unitarity condition and the
Yang-Baxter equations:
\begin{equation}
\label{unitary}
  Y^{ab}_{jk}\ Y^{ab}_{kj}=1
\end{equation}
\begin{equation}
  \label{YB1}
 Y_{jk}^{ab}\  Y_{ik}^{bc}\   Y_{ij}^{ab} =
 Y_{ij}^{bc}\  Y_{ik}^{ab}\   Y_{jk}^{bc} .
\end{equation}

In the following we attach the rapidities to the vector spaces. 
To every permutation of the rapidities $Q\in S_N$ we associate a
configuration
\begin{equation}
  V^{(Qp)_1}_1 \otimes   V^{(Qp)_2}_2 \otimes \cdots \otimes   V^{(Qp)_N}_N  .
\end{equation}

We define an operator $\mathcal{Q}(Q,\{p\}):\ (S_N\times \complex^N)\to \text{End}(V^{(N)})$ as
follows. The permutation $Q\in S_N$ is 
re-constructed as a product of elementary permutations and to  every
exchange of rapidities 
\begin{equation}
    V^{p_j}_a \otimes   V^{p_k}_b\quad\to\quad    V^{p_k}_a \otimes   V^{p_j}_b
\end{equation}
we associate the action of $Y^{ab}_{jk}$;
the operator $\mathcal{Q}(Q,\{p\})$ is defined as the product of the $Y^{ab}_{jk}$ matrices.
 This
definition leads to the property
\begin{equation}
\label{CNYANG}
\mathcal{Q}(Q_2,Q_1p)\
\mathcal{Q}(Q_1,p)=\mathcal{Q}(Q_2Q_1,p).
\end{equation}
The consistency of the construction is guaranteed by the Yang-Baxter
equation \eqref{YB1}. 

\bigskip

Now we are in a position to construct the vector valued wave functions:
\begin{equation*}
  \chi_N:\valos^N\to  V^{(N)}.
\end{equation*}
The wave functions depend on the (ordered)
set of rapidities $\{p\}_N$ and an arbitrary (fixed) vector $\omega_N\in
V^{(N)}$, which 
is a parameter describing the polarization of the wave function. It
will be specified in the two-component case in section \ref{2c}.

We define the fundamental domain as 
\begin{equation}
  x_j>x_l\quad\text{iff}\quad j>l.
\end{equation}
In this region the wave function is
\begin{equation}
\label{dejaren}
  \chi_N(x|p,\omega_N)=\frac{1}{\sqrt{N!}} \sum_{Q\in S_N}
  e^{i\skalarszorzat{Qp}{x}}  \mathcal{Q}(Q,p)\omega_N.
\end{equation}
It can be extended to $\valos^N$ by symmetry. To write down the formal
relation we need the representation $\rho(Q)$ of the permutation $Q$
which is such that
\begin{equation*}
  (Qx)_1<\dots < (Qx)_N.
\end{equation*}
Then the wave function in $\valos^N$ reads
\begin{equation}
\label{dejarenX}
  \chi_N(x|p,\omega_N)=\frac{1}{\sqrt{N!}} 
\mathcal{\rho}(Q^{-1})
\sum_{R\in S_N}
  e^{i\skalarszorzat{Rp}{Qx}}  \mathcal{Q}(R,p)\omega_N.
\end{equation}

It can be shown that the
wave function defined this way is an eigenstate of the Hamiltonian
\eqref{HM-LL} for arbitrary $\{p\}_N$ and $\omega_N$. The total energy and
the momentum is given by
\begin{equation*}
  E_N=\sum_j p_j^2\qquad\qquad P_N=\sum_j p_j,
\end{equation*}
the vector $\omega_N$ only determines the polarization of the wave
function.

We wish to note that the matrix
\begin{equation}
  X^{ab}_{jk}=P_{ab} Y^{ab}_{jk}=
\frac{\sigma (p_j-p_k)-icP_{ab}}{p_j-p_k+ic}
\label{X}
\end{equation}
can be interpreted as the two-particle S-matrix of the
theory. Therefore, the individual coefficients in \eqref{dejarenX}
describe the two-particle scattering events.

It is important to establish 
 the exchange property of the wave function with respect to an
exchange of rapidities:
\begin{thmlem}
  \begin{equation}
    \chi_N(x|p_1,\dots,p_{j},p_{j+1},\dots,p_N,\omega_N)=
 \chi_N(x|p_1,\dots,p_{j+1},p_{j},\dots,p_N,Y^{j,j+1}_{j,j+1}
 \omega_N).
\label{chiNex}
  \end{equation}
\end{thmlem}
\begin{proof}
It is enough to consider the fundamental region. 
  Let us denote the exchange $(j\leftrightarrow j+1)$ by $P_j$. 
We introduce a new summation variable $Q'=QP_j$ leading to
  \begin{equation*}
      \chi_N(x|p_1,\dots,p_{j},p_{j+1},\dots,p_N,\omega_N)=
\frac{1}{\sqrt{N!}} \sum_{Q'\in S_N}
  e^{i\skalarszorzat{Q'P_j p}{x}}  \mathcal{Q}(Q'P_j,p)\omega_N.
  \end{equation*}
It follows from \eqref{CNYANG} that
\begin{equation*}
  \mathcal{Q}(Q'P_j,p)=  \mathcal{Q}(Q',P_j p) 
  \mathcal{Q}(P_j,p)=
 \mathcal{Q}(Q',P_j p)  Y^{j,j+1}_{j,j+1}.
\end{equation*}
Therefore
  \begin{equation*}
\begin{split}
      \chi_N(x|p_1,\dots,p_{j},p_{j+1},\dots,p_N,\omega_N)&=
\frac{1}{\sqrt{N!}} \sum_{Q'\in S_N}
  e^{i\skalarszorzat{Q'P_j p}{x}}  \mathcal{Q}(Q',P_jp)
  Y^{j,j+1}_{j,j+1}\omega_N=\\
&=  \chi_N(x|p_1,\dots,p_{j+1},p_{j},\dots,p_N,
   Y^{j,j+1}_{j,j+1} \omega_N).
\end{split}
  \end{equation*}
\end{proof}

\subsection{The Form Factors}

\label{MFF}

We are interested in the form factors of the field operators
$\Psi_l(0)$
and the bilinear operators
$\rho_{jl}(0)=\Psi^\dagger_j(0)\Psi_l(0)$. The latter encompass the
density operators of particles with a given spin (when $j=l$) and also
the spin-flip operators (when $j\ne l$).

The (infinite volume) form factors are co-vector valued functions of
the rapidities. Evaluated on two vectors $\psi_N\in
V^{(N)}$ and $\phi_{N+1}\in V^{(N+1)}$ they are
given 
as the coordinate space integrals
\begin{equation}
\label{FinfdefM}
\begin{split}
&  \mathcal{F}^l_N(\{p\}_N,\{k\}_{N+1})(\psi_N,\phi_{N+1})=
\lim_{\eps\to 0}\sqrt{N+1} \int_{-\infty}^\infty  dx_1 \dots dx_N\
\prod_{j=1}^N f_\eps(x_j)
\times \\
&\hspace{2cm}\Big\langle \chi_N(x_1,\dots,x_N|\{p\}_N,\psi_N)   \Big|
U^{(0)}_l\chi_{N+1}(0,x_1,\dots,x_N|\{k\}_{N+1},\phi_{N+1})\Big\rangle_N
\end{split}
\end{equation}

\begin{equation}
\begin{split}
\label{Finfdef2M}
&  \mathcal{F}^{jl}_N(\{p\}_N,\{k\}_{N})(\psi_N,\phi_{N+1}) =
\lim_{\eps\to 0}
N \int_{-\infty}^\infty dx_1 \dots dx_{N-1}\  \prod_{j=1}^{N-1} f_\eps(x_j) \times \\
&\hspace{2cm}\Big\langle U^{(1)}_j\chi_N(p|0,x_1,\dots,x_{N-1},\psi_N)\Big|
U^{(1)}_l \chi_{N}(k|0,x_1,\dots,x_{N-1},\phi_{N})\Big\rangle_{N-1}.
\end{split}
\end{equation}
Here $U^{(j)}_l$ is an operator 
\begin{equation}
\label{uudef}
  U^{(j)}_l:V^{(N)}\to V^{(N-1)}
\end{equation}
which acts by taking the scalar product with the unity vector $e_l$ in the
$j$-th vector space and leaving the others invariant:
\begin{equation*}
\begin{split}
 & U^{(j)}_l\Big(e_{a_1}\otimes e_{a_2} \otimes\dots\otimes
e_{a_{j-1}}\otimes  e_{a_j}\otimes e_{a_{j+1}}\otimes \dots \otimes e_{a_N}\Big)
=\\
&\hspace{2cm}=\delta_{l,a_j} (e_{a_1}\otimes e_{a_2} \otimes\dots\otimes
e_{a_{j-1}}\otimes e_{a_{j+1}}
\otimes \dots \otimes e_{a_N})
\end{split}
\end{equation*}

The scalar
products in \eqref{FinfdefM}-\eqref{Finfdef2M} are the canonical ones
in $V^{(N)}$ and $V^{(N-1)}$, respectively. In order to conform with
our previous notations, in the case of the field operator the
indexation of the vector spaces corresponding to the rapidities
$\{k\}_{N+1}$ is given by
\begin{equation}
 V^{(N+1)}= V^{k_0}_0 \otimes   V^{k_1}_1 \otimes \cdots \otimes   V^{k_N}_N  .
\end{equation}

In the following we establish the analytic properties of the form
factors. 
The behaviour under the exchange of rapidities follows simply from the
properties \eqref{chiNex} of the Bethe Ansatz wave functions:
\begin{equation}
\label{Mex}
\begin{split}
&
\mathcal{F}_N^l(p_1,\dots,p_N|k_0,\dots,k_j,k_{j+1},\dots,k_N)(\phi_N,\psi_{N+1})
=\\
&\hspace{3cm}F_N^l(p_1,\dots,p_N|k_0,\dots,k_{j+1},k_{j},\dots,k_N)
(\phi_N,\hat S_{j,j+1} \psi_{N+1})
\\
&  \mathcal{F}_N^l(p_1,\dots,p_j,p_{j+1},\dots,p_N|k_0,\dots,k_N)(\phi_N,\psi_{N+1})=\\
& \hspace{3cm}F_N^l(p_1,\dots,p_{j+1},p_{j},\dots,p_N|k_0,\dots,k_N)
(\hat S^{j,j+1}\phi_N,\psi_{N+1}).
\end{split}
\end{equation}
Here we introduced the short-hand notation
\begin{equation}
  \label{hatSdef}
\hat S_{j,l}=Y_{jl}^{jl}(\{k\})
=\frac{(k_j-k_l)\rho_{jl}-ic}{k_j-k_l+ic}
\qquad\qquad
\hat S^{j,l}=Y_{jl}^{jl}(\{p\})=
\frac{(p_j-p_l)\rho_{jl}-ic}{p_j-p_l+ic}.
\end{equation}
Analogous relations hold for the bilinear operators as well.

The singularity properties of the form factors are more
involved. All the poles of the form factors arise from the coordinate
space integrals and therefore the positions of the poles are identical
with those in the scalar case, ie. the form factors have poles at
$p_j=k_l$,   whenever two rapidities from the two sides
coincide. To write down the residues we introduce the following notation: 
\begin{equation*}
  Id_{1,0}\otimes \mathcal{F}_{N-1}\qquad \text{and}\qquad
\mathcal{F}_{N-1}\otimes Id_{N,N}
\end{equation*}
are operations where it is understood that a trace is taken with
respect to the corresponding spaces of $\psi_N$ and $\phi_{N+1}$ and
the remaining vector spaces are substituted in the form factor. For
example
\begin{equation*}
  \begin{split}
&( Id_{1,0}\otimes \mathcal{F}_{N-1})
(e_{a_1}\otimes \dots \otimes e_{a_N},e_{b_0}\otimes \dots \otimes
e_{b_{N}})=\\
&\hspace{2cm}\delta_{a_1,b_0} \mathcal{F}_{N-1}(e_{a_2}\otimes \dots \otimes e_{a_N},e_{b_1}\otimes \dots \otimes
e_{b_{N}})
  \end{split}
\end{equation*}

With
this notation the kinematical poles are given by the following theorem.

\begin{thm}
The pole of the field operator form factor at $p_N=k_N$ is given by
\begin{equation}
  \begin{split}
\label{infect-me}
&\mathcal{F}^{l}_N(p_1,\dots,p_{N}|k_0,\dots,k_{N})(\psi_N,\phi_{N+1})\sim
\frac{i}{k_N-p_N}\times\\
&\left[(\mathcal{F}^l_{N-1}\otimes Id_{N,N})(\psi_N,\phi_{N+1}) 
-\sigma (Id_{1,0}\otimes \mathcal{F}^l_{N-1})
(\hat S^{1,N}\dots \hat S^{N-1,N} \psi_N,
\hat  S_{0,N}\dots \hat S_{N-1,N}\ \phi_{N+1})
 \right].
  \end{split}
\end{equation}
Here we abbreviated $\mathcal{F}^l_{N-1}=\mathcal{F}^l_{N-1}(p_1,\dots,p_{N-1}|k_0,\dots,k_{N-1})$.
 Poles at other $p_j=k_l$ can be obtained using the exchange property.  
\end{thm}
\begin{proof}
We follow the ideas of the proof of theorem \ref{hukko}.
We consider those terms where $p_N$ and $k_N$ are coupled to $x_N$
and there is a singularity. Similar to the scalar case, the only two
possibilities are if $x_N$ is larger, or smaller than any of the other
coordinates. To be specific we first consider the following two cases:
\begin{equation*}
  x_1<\dots < x_{N-1}<x_N,\ 0<x_N\qquad\text{and} \qquad 
x_N< x_1<\dots < x_{N-1},\ x_N<0.
\end{equation*}
In the first case \eqref{dejaren} gives no action on the two vectors,
therefore this term yields
\begin{equation}
\label{oneof1}
  \frac{i}{k_N-p_N}   \mathcal{F}_{N-1}^l \otimes Id_{N,N}.
\end{equation}
In the second case the wave function is obtained by symmetrization
from \eqref{dejaren}:
\begin{equation*}
  \chi_{N+1}(0,x_1,\dots,x_N|\{k\})=\sigma^{N}P_{N-1,N}\dots P_{01}
  \chi_{N+1}(x_N,0,x_1,\dots,x_{N-1}|\{k\}).
\end{equation*}
Here we used
\begin{equation*}
  (P_{01}P_{12}\dots P_{N-1,N})^{-1}= P_{N-1,N} \dots P_{01}.
\end{equation*}
To obtain the coefficient of the exponential
\begin{equation*}
  e^{i(k_1x_1+\dots+k_Nx_N)}
\end{equation*}
we have to consider the permutation $Q\in S_{N+1}$
\begin{equation*}
  Q(\{k_0,\dots,k_{N}\})=\{k_N,k_0,\dots, k_{N-1}\}
\end{equation*}
For this permutation the corresponding linear operator is
\begin{equation*}
  \mathcal{Q}(Q,\{k\})=\hat S_{0,N}\dots \hat S_{N-1,N}.
\end{equation*}
Performing similar steps for the dual vector we obtain the scalar
product
\begin{equation*}
\sigma  \Big\langle
P_{N-1,N}\dots P_{1,2} \hat S^{1,N}\dots \hat S^{N-1,N} \phi_N
|
U^{(0)}_lP_{N-1,N}\dots P_{12}P_{0,1} 
\hat S_{0,N}\dots \hat S_{N-1,N}
\phi_{N+1}
\Big\rangle_N.
\end{equation*}
This is equivalent to
\begin{equation*}
\sigma  \Big\langle
P_{12}\dots P_{N-1,N} \hat S^{1,N}\dots \hat S^{N-1,N} \phi_N
|
P_{12}\dots P_{N-1,N} U^{(1)}_l
\hat S_{0,N}\dots \hat S_{N-1,N}
\phi_{N+1}
\Big\rangle_N.
\end{equation*}
Moreover the scalar product is invariant with the permutation of
vector spaces within $V^{(N)}$ therefore the above scalar product is
equivalent to
\begin{equation*}
\sigma  \Big\langle
 \hat S^{1,N}\dots \hat S^{N-1,N} \phi_N
|
 U^{(1)}_l
\hat S_{0,N}\dots \hat S_{N-1,N}
\phi_{N+1}
\Big\rangle_N.
\end{equation*}
For the contribution in question the above scalar product  is
equivalent to the action of the operator 
\begin{equation}
\label{oneof2}
-\frac{i\sigma}{k_N-p_N}  (Id_{1,0}\otimes \mathcal{F}_{N-1}^l)
\left(\hat S^{1,N}\dots \hat S^{N-1,N} \phi_N,
\hat S_{0,N}\dots \hat S_{N-1,N}
\phi_{N+1}
\right).
\end{equation}
Adding the contributions \eqref{oneof1}-\eqref{oneof2}, performing the
summations over the remaining possibilities for the permutations, and
using the Yang-Baxter equation one obtains finally the statement \eqref{infect-me}.
\end{proof}

The residue equation \eqref{infect-me} involves the parameter $\sigma$
which distinguishes the bosonic and fermionic cases. It is useful to
write down a relation which does not depend on $\sigma$. Therefore we
introduce one more $S$-matrix, namely
\begin{equation}
  \label{Sezis}
  \tilde S_{j,k}=\sigma \hat S_{j,k}\qquad\qquad
 \tilde S^{j,k}=\sigma \hat S^{j,k}.
\end{equation}
Then the kinematical pole equation reads
\begin{equation}
\begin{split}
\label{infect-me2}
&\mathcal{F}^{l}_N(p_1,\dots,p_{N}|k_0,\dots,k_{N})(\psi_N,\phi_{N+1})\sim
\frac{i}{k_N-p_N}\times\\
&\left[(\mathcal{F}^l_{N-1}\otimes Id_{N,N})(\psi_N,\phi_{N+1}) 
- (Id_{1,0}\otimes \mathcal{F}^l_{N-1})
(\tilde S^{1,N}\dots \tilde S^{N-1,N} \psi_N,
\tilde  S_{0,N}\dots \tilde S_{N-1,N}\ \phi_{N+1})
 \right].
  \end{split}
\end{equation}
It can be shown that an analogous equation (with the $\tilde S$
operators involved) holds for the bilinear
operators and arbitrary higher body local operators as
well, irrespective of the statistics of the model. As a final remark we note that equation
\eqref{infect-me2} can be considered as a 
non-relativistic version of the kinematical pole axiom known in the
form factor bootstrap in integrable relativistic QFT's \cite{smirnov_ff}.

We conjecture that in the case of the field operator the recursion
relation \eqref{infect-me} together with the exchange properties
\eqref{Mex} determine the form factors uniquely. Then, at least in
principle they can be
constructed with a procedure similar to the one presented in
\cite{Balog:jopofa} for the case of the relativistic $O(3)$
$\sigma$-model. We leave this problem to further research.

In the case of the bilinear operators it is expected that the
recursion relations are not restrictive enough. However, similar to
the one-component case there is a useful asymptotic condition:
\begin{thm}
  The asymptotic behaviour of the form factors of the bilinear
  operators is given by
  \begin{equation}
    \label{elmegy}
\begin{split}
&    \lim_{p_1\to \infty}
\mathcal{F}^{jk}_N(p_1,\dots,p_N|k_1,\dots,k_N)(\psi_N,\phi_N)=\\
&\hspace{3cm}\mathcal{F}^{k}_{N-1}(p_2,\dots,p_N|k_1,\dots,k_N)(U^{(1)}_j\psi_N,\phi_N).
\end{split}
  \end{equation}
\end{thm}
\begin{proof}
  The statement of the theorem is a direct consequence of the relation
  \begin{equation}
\label{ezalimit2}
\begin{split}
& \lim_{p_1\to\infty} \Psi_j(0)\chi_N(0,x_1,\dots,x_{N-1}|p_1,\dots,p_N,\omega_N)\sim\\
&\hspace{3cm}\chi_{N-1}(x_1,\dots,x_{N-1}|p_2,\dots,p_N,U^{(1)}_j\omega_N),
\end{split}
  \end{equation}
where similar to \eqref{ezalimit} the sign $\sim$ indicates that on
the l.h.s. only those terms are kept where $p_1$ is attached to
$x_0=0$. Equation \eqref{ezalimit2} is checked easily using the
definitions \eqref{dejaren}-\eqref{dejarenX} and the the limiting values
\begin{equation*}
\lim_{p_j\to\infty} Y_{jk}^{ab}=\rho_{ab}.
\end{equation*}
\end{proof}

We conjecture that the form factors of the bilinear operators are
determined uniquely by the exchange relations, the kinematical pole
axiom \eqref{infect-me2} and the asymptotic condition
\eqref{elmegy}. Also, we note that \eqref{elmegy} can be considered as
a non-relativistic version of the factorization property known in
relativistic integrable QFT \cite{Delfino:1996nf}. 

\section{Two-component systems: the nested Bethe Ansatz in infinite volume}

\label{2c}

In this section we consider the two-component models.
As a first step we construct the so-called nested Bethe Ansatz states.
Our approach is somewhat different from the usual one: the results
presented here apply directly in infinite volume, therefore we are not
concerned with the periodicity of the wave functions. 
The
connection to the finite volume states is made in section \ref{fftcsa1hehe}.

As a second step we also introduce the 
``magnonic form factors'' which are matrix elements of local operators
on the nested BA states. We investigate the analytic properties of
these objects and obtain a set of ``magnonic form factor equations''.

\subsection{The nested BA states}

We consider the infinite volume Bethe Ansatz states defined in \eqref{dejaren} in the
case of $V=\complex^2$. The basis in $V$ is formed by the two vectors
$\ket{+}$ and $\ket{-}$. 
Our aim here is to specify the vectors $\omega_N$ entering the Bethe
Ansatz states. 

We consider an ``auxiliary space'' $V_a=\complex^2$ and the operator (also
called the ``monodromy matrix'')
\begin{equation}
\label{T1}
  T(u|p)=X_{aN}(u-p_N)\dots X_{a1}(u-p_1).
\end{equation}
The trace in auxiliary space is called the transfer matrix:
\begin{equation}
\label{T2}
  t(u|p)=\text{Tr}_a   T(u|p).
\end{equation}
The operator \eqref{T1} can be viewed as the monodromy matrix of an
inhomogeneous spin chain, where the rapidities $p_j$ play the role of
inhomogeneities. Then the standard Algebraic Bethe Ansatz techniques
can be used to construct Bethe states in $V^{(N)}$. 

To establish the notations we recall that
 the rational $R$-matrix is defined as
\begin{equation}
  \label{R}
  R(u,c)=\frac{1}{u+ic}
  \begin{pmatrix}
    u+ic & & & \\
& u & ic & \\
& ic & u & \\
& & & u+ic\\
  \end{pmatrix}.
\end{equation}
Comparing to the formula \eqref{X} note that
\begin{equation}
\label{best}
  X(u)=S_{1}(u,c)  R(u,-\sigma c)\qquad\qquad S_{1}(u,c)=\frac{\sigma u -ic}{u+ic}.
\end{equation}
Here $S_{1}(u)$ is the ``one-particle'' S-matrix, which describes the
amplitude associated to the exchange of two particles with the same
spin. In the fermionic case it is equal to $(-1)$, whereas in the
bosonic case it is just the Lieb-Liniger amplitude. 

In order to conform with the conventions of the spin chain literature we 
define the normalized monodromy matrix
\begin{equation}
   \tilde  T(u|p)=R_{aN}(u-p_N)\dots R_{a1}(u-p_1).
\end{equation}
Here it is understood that in the $R$-operators the coupling constant
is $-\sigma c$.
With respect to the auxiliary space it is written  as
\begin{equation*}
  \tilde T(u|p)=
  \begin{pmatrix}
    \tilde A(u|p) & \tilde B(u|p) \\
\tilde C(u|p) & \tilde D(u|p)
  \end{pmatrix}.
\end{equation*}
The commutation relations of the elements of the monodromy matrix can
be expressed in the compact form
\begin{equation}
\label{RTT}
  R(u-v)\tilde T(u)\otimes \tilde T(v)=\tilde T(v)\otimes \tilde T(u) R(u-v),
\end{equation}
which is understood as an equation in the tensor product of two
auxiliary spaces. Eq. \eqref{RTT} follows from a repeated use of the Yang-Baxter
equations \eqref{YB1}. It follows from \eqref{RTT} that the transfer
matrices $t(u|p)$ form a commuting set of operators.

We fix the reference state $\ket{+}_N=\ket{++\dots +}_N\in V^{(N)}$. 
The $B(\mu)$-operators can be considered as creation operators of (interacting)
spin-waves of $\ket{-}$ spins acting on the reference state. The parameter $\mu$  describes the
rapidity of the spin wave and is often called the magnonic rapidity.

We 
define the function $\omega_N(\{p\}_N,\{\mu\}_M): \complex^N\times
\complex^M\to V^{(N)}$ as
\begin{equation}
\label{omegadef}
  \omega_N(\{p\}_N,\{\mu\}_M)=\tilde B(\mu_1+i\sigma c/2|\{p\})
\dots \tilde B(\mu_M+i\sigma c/2 |\{p\}) \ket{+}_N.
\end{equation}
The shift of $i\sigma c/2$ is introduced for technical
reasons. According to \eqref{RTT} the $B$-operators commute with each
other, 
therefore the function $\omega_N(\{p\}_N,\{\mu\}_M)$ is completely symmetric
with respect to the set $\{\mu\}$. The exchange
properties with respect to the set $\{p\}$ are determined once more by
the Yang-Baxter relation:

\begin{thmlem}
  The function $\omega_N(\{p\}_N,\{\mu\}_M)$ satisfies
\begin{equation}
\label{goed}
\omega_N(p_1,\dots,p_{j+1},p_j,\dots,p_N,\{\mu\})=
\hat R_{j,j+1} \omega_N(p_1,\dots,p_{j},p_{j+1},\dots,p_N,\{\mu\}),
\end{equation}
where $\hat R_{j,j+1}=P_{j,j+1}R_{j,j+1}(p_j-p_{j+1})$.
\end{thmlem}
\begin{proof}
A modified form of the Yang-Baxter equation \eqref{YB1} is
\begin{equation}
\label{YXX}
\begin{split}
&  Y_{j,j+1}(p_j-p_{j+1}) X_{j+1,a}(u-p_{j+1}) X_{ja}(u-p_j)
=\\&\hspace{3cm} X_{j+1,a}(u-p_j) X_{ja}(u-p_{j+1}) Y_{j,j+1}(p_j-p_{j+1}) .
\end{split}
\end{equation}
This implies
\begin{equation}
\label{goed2}
B(\mu|p_1,\dots,p_{j+1},p_j,\dots,p_N) \hat R_{j,j+1}=
\hat R_{j,j+1} B(\mu|p_1,\dots,p_{j},p_{j+1},\dots,p_N).
\end{equation}
By commuting $\hat R$ through the $B$-operators and using the fact
that $\hat R$ acts trivially on the reference state we 
obtain the statement \eqref{goed}.
\end{proof}

This leads immediately 
to the following lemma:
\begin{thmlem}
\label{thmlem1}
  The states defined by \eqref{omegadef} satisfy the exchange property
  \begin{equation}
\mathcal{Q}(P,p)   \omega_N(\{p\}_N,\{\mu\}_M)=
\mathop{\prod_{j<l}}_{Pj>Pl} S_{1}(p_j-p_l) \times
\omega_N(\{Pp\}_N,\{\mu\}_M)
  \end{equation}
for arbitrary $P\in S_N$.
\end{thmlem}
\begin{proof}
It is enough to check the statement for the elementary
permutations. Then the statement follows from equations \eqref{goed}
and \eqref{best}.
\end{proof}

Finally we obtain the main statement about the nested Bethe Ansatz states:
\begin{thm}
  The coordinate space wave functions defined as
  \begin{equation}
\label{hukkle}
      \chi_N(x|\{p\}_N,\{\mu\}_M)=\frac{1}{\sqrt{N!}} \sum_{P\in S_N}
  e^{i\skalarszorzat{Pp}{x}}  
\mathop{\prod_{j<l}}_{Pj>Pl} S_{1}(p_j-p_l) 
\omega_N(\{Pp\}_N,\{\mu\}_M),
  \end{equation}
where $\omega_N(\{p\}_N,\{\mu\}_M)$ is given by \eqref{omegadef}, are
eigenstates of the Hamiltonian \eqref{HM-LL}. Here it is understood
that \eqref{hukkle} is defined in the fundamental domain
$x_1<\dots<x_N$ and it is extended to the other domains by symmetry.
\end{thm}
\begin{proof}
  The statement follows from Lemma \ref{thmlem1} and the fact that the states
  \eqref{dejaren} are eigenstates.
\end{proof}

It is easy to see that the wave function \eqref{hukkle} possesses the
following exchange property with respect to the set $\{p\}$:
\begin{equation}
\label{LLex}
\begin{split}
&\chi_N(x|p_1,\dots,p_j,p_{j+1},\dots,p_N,\{\mu\}_M)=\\
&\hspace{4cm}S_1(p_j-p_{j+1})\chi_N(x|p_1,\dots,p_{j+1},p_{j},\dots,p_N,\{\mu\}_M).
\end{split}
\end{equation}
On the other hand, $\chi_N$ is completely symmetric with respect to the sets $\{\mu\}$.

It is useful to obtain explicit representations for the
vectors
 $\omega_N(\{p\}_N,\{\mu\}_M)$.
Here we just present the known results \cite{iz-kor-resh} using the notations of
 \cite{HubbardBook}:
\begin{equation}
\label{intreatment}
  \omega_N(\{p\}_N,\{\mu\}_M)=\sum_{a_1,\dots,a_M=1}^N
A(a_1,\dots,a_M)
\sigma_-^{a_1}\dots \sigma_-^{a_M}
\ket{+},
\end{equation}
where
\begin{equation}
\label{Afunct}
  A(a_1,\dots,a_M)=\frac{1}{M!}\sum_{R\in S_M}
\prod_{1\le k<l\le M}\frac{(R\mu)_l-(R\mu)_k+i\sigma  c\epsilon(a_l-a_k)}
{(R\mu)_l-(R\mu)_k}
\prod_{l=1}^M \mathcal{A}((R\mu)_l,a_l),
\end{equation}
where $\epsilon(a)$ is the sign function and the propagator of the spin waves is given by
\begin{equation}
  \mathcal{A}(u,a)=\frac{-i\sigma c}{u-p_a-i\sigma c/2}
  \prod_{b=1}^{a-1} \frac{u-p_b+i\sigma c/2}{u-p_b-i\sigma c/2}.
\end{equation}
With this we have finished the explicit construction of the nested
Bethe Ansatz states. Note that we did not address the completeness of
states; in the present approach the magnonic rapidities are arbitrary
parameters, they are not assumed to satisfy the Bethe equations. 

It is useful to consider the $\mu\to\infty$ limit of the
wave function. We obtain the following statement:
\begin{thmlem}
\label{we}
Sending a magnonic rapidity to infinity yields the action of the
overall spin lowering operator:
  \begin{equation}
    \label{peyote}
\lim_{\mu_1\to\infty}\   \mu_1\   \chi_N(x|\{p\}_N,\mu_1,\dots,\mu_M)=
-i\sigma c S_-\ \chi_N(x|\{p\}_N,\mu_2,\dots,\mu_M).
  \end{equation}
\end{thmlem}
\begin{proof}
  It is known from the Algebraic Bethe Ansatz that
  \begin{equation}
\label{Blim}
    \lim_{\mu\to\infty}\ \mu\ \tilde B(\mu|\{p\}_N)=-i\sigma c S_-. 
  \end{equation}
This is easily seen from the construction of the transfer matrix
\eqref{T1} or from the explicit expression \eqref{intreatment}.
Applying \eqref{Blim} to \eqref{omegadef} and finally to
\eqref{hukkle} we obtain the statement of the lemma.
\end{proof}

\subsection{Magnonic Form Factors}

The magnonic Form Factors are defined as the matrix elements of local
operators on the nested Bethe states. They are scalar
functions of four sets of rapidities. 

In the case of the field operator the form factor is defined as
\begin{equation}
\begin{split}
  \mathcal{F}^l_N(\{p\}_N,\{\nu\}_{M'},\{k\}_{N+1},\{\mu\}_M)=
\lim_{\eps\to 0}\sqrt{N+1} \int_{-\infty}^\infty  dx_1 \dots dx_N\
\prod_{j=1}^N f_\eps(x_j)\\
\times \Big\langle \chi_{N}(x_1,\dots,x_N|\{p\},\{\nu\})
\Big| U^{(0)}_l\chi_{N+1}(0,x_1,\dots,x_N|\{k\},\{\mu\})
\Big\rangle_N.
\end{split}
\end{equation}
Due to spin conservation 
\begin{equation*}
  M'=M+\frac{l-1}{2}
\end{equation*}

In the case of the bilinear operators the form factors are defined as
\begin{equation}
\begin{split}
  \mathcal{F}^{lm}_N(\{p\}_N,\{\nu\}_{M'},\{k\}_{N},\{\mu\}_M)=
\lim_{\eps\to 0}N \int_{-\infty}^\infty  dx_1 \dots dx_{N-1}\
\prod_{j=1}^{N-1} f_\eps(x_j)\\
\times \Big\langle U^{(1)}_l\chi_{N}(0,x_1,\dots,x_{N-1}|\{p\},\{\nu\})
\Big| U^{(1)}_m\chi_{N}(0,x_1,\dots,x_{N-1}|\{k\},\{\mu\})
\Big\rangle_{N-1}.
\end{split}
\end{equation}
Spin conservation requires
\begin{equation*}
  M'=M+\frac{m-l}{2}.
\end{equation*}

The $U^{(j)}_l$ operators entering the formulas above are the projectors introduced
in \eqref{uudef}.

\bigskip

It is important to establish the analytic structure of the magnonic
form factors.

It follows from the properties of the nested wave functions
that the form factors are completely symmetric with
respect to the magnonic rapidities. With respect to the particle
rapidities they possess the exchange property \eqref{LLex} (and its
complex conjugate).

Concerning the kinematical poles we substitute 
\eqref{omegadef} into
 \eqref{infect-me} 
and  using
Lemma \ref{thmlem1} we obtain
\begin{equation}
  \begin{split}
\label{infect-me-a-kurva-eletbe}
&\mathcal{F}^{l}_N(\{p\}|\{k\})\Big(\omega_N(p_1,\dots,p_N,\{\nu\}),\omega_{N+1}
(k_0,\dots,k_N,\{\mu\})\Big)\sim
\frac{i}{k_N-p_N}\times\\
&\left[(\mathcal{F}^l_{N-1}\otimes Id_{N,N})
\Big(\omega_N(p_1,\dots,p_N,\{\nu\}),\omega_{N+1}(k_0,\dots,k_N,\{\mu\})\Big)
-\right.\\
&- \prod_{j=0}^{N-1} S_1(k_j-k_N)   \prod_{j=1}^{N-1} S_1(p_N-p_j)  \times\\
&\left.\times\sigma (Id_{1,0}\otimes \mathcal{F}^l_{N-1})
\Big(\omega_N(p_N,p_1,\dots,p_{N-1},\{\nu\}),\omega_{N+1}(k_N,k_0,\dots,k_{N-1},\{\mu\})\Big)
 \right].
  \end{split}
\end{equation}
Note that in both terms the $Id$ operator 
acts on those vector spaces to which
 the rapidities $p_N$ and $k_N$ are attached. 
Taking the scalar product with respect to these spaces 
leads to two possibilities: Either the
corresponding components are $+$ or $-$. 
In both cases the resulting
amplitudes are evaluated easily using the explicit representation
\eqref{intreatment}. We arrive at the following ``inhomogeneous''
form factor recursion equation:
\begin{equation}
  \label{psybient}
\begin{split}
 &
 \mathcal{F}_N^l\big(\{p\}_{1..N},\{\nu\}_{1..M'}|\{k\}_{0..N},\{\mu\}_{1..M}\big)\sim
\frac{i}{k_N-p_N}\times \\
&\times\left[1-\prod_{t=1}^M S_{\frac{1}{2}}(\mu_t-k_{N})\prod_{t=1}^{M'} S_{\frac{1}{2}}(p_N-\nu_t)
\prod_{j=0}^{N-1}\tilde S_1(k_{jN})\prod_{k=1}^{N-1}\tilde S_1(p_{Nk})
\right] \times \\
&\times \mathcal{F}^l_{N-1}(\{p\}_{1..N-1},\{\nu\}_{1..M'}|\{k\}_{0..N-1},\{\mu\}_{1..M})\\
&+\sum_{s=1}^{M'}\sum_{r=1}^M \frac{c^2}{(\mu_r-k_N-i\sigma c/2)(\nu_s-p_N+i\sigma
  c/2)}\times \\
&\left[
\prod_{j=0}^{N-1} S_{\frac{1}{2}}(\mu_r-k_j)\prod_{j=1}^{N-1} S_{\frac{1}{2}}(p_j-\nu_s)
\mathop{\prod_{t=1}^M}_{t\ne r} W(\mu_t-\mu_r)
\mathop{\prod_{t=1}^{M'}}_{t\ne s} W(\nu_s-\nu_t)-\right.\\
&-\left.  \mathop{\prod_{t=1}^M}_{t\ne r} S_{\frac{1}{2}}(\mu_t-k_{N}) W(\mu_r-\mu_t)
\mathop{\prod_{t=1}^{M'}}_{t\ne s} S_{\frac{1}{2}}(p_N-\nu_t) W(\nu_t-\nu_s)
\prod_{j=0}^{N-1}\tilde S_1(k_{jN})\prod_{k=1}^{N-1}\tilde S_1(p_{Nk})
\right]\\
& \times
  \mathcal{F}^l_{N-1}(\{p\}_{1..N-1},\{\nu\}_{1..\hat s..M'}|\{k\}_{0..N-1},\{\mu\}_{1..\hat r..M}).
\end{split}
\end{equation}
Here we used
\begin{equation*}
\tilde S_1(u)=\sigma S_1(u)=\frac{u-i\sigma c}{u+ic}\qquad\qquad
  S_{\frac{1}{2}}(u)=\frac{u+i\sigma c/2}{u-i\sigma c/2}\qquad\qquad
W(u)=\frac{u+i\sigma c}{u},
\end{equation*}
and in the last line it is understood that the magnonic rapidities
$\nu_s$ and $\mu_r$ are not substituted into the form factor.

Analogous relations can be written down for the bilinear operators as
well. In those cases the ranges for the different products change according to
the number of rapidities present.

The interpretation of the residue equation \eqref{psybient} is the
following. The kinematical poles arise when two particle rapidities approach
each other; in the multi-component case this leads to different
contributions corresponding to the 
different  spin orientations. In the cases when this component
is $\ket{+}$ the remaining part of the wave function has the same number of
$\ket{-}$ spins, and its explicit polarization is described by the
same sets of magnonic rapidities. This corresponds to the second and
third lines of  \eqref{psybient}. On the other hand, when the singular
piece of the wave function carries a $\ket{-}$ spin, this is
associated with one of the magnonic rapidities from both
states. Therefore the remaining part of the wave functions yields 
form factors with one less number of $\ket{-}$ spins, ie. one
less magnonic rapidity, just as in the form factors on the seventh line
of  \eqref{psybient}.

The above recursive equations can be called ``inhomogeneous'' in the
sense that they can not be solved by considering fixed sets of
$\{\mu\}$ and $\{\nu\}$ as spectator variables: any
recursion procedure with a given number of magnonic rapidities
involves all the form factors with one less magnonic rapidity.
This makes the solution of
the system \eqref{psybient} more involved. 

Note that the recursion relation \eqref{psybient} has a different
structure than the singularity properties of the scalar products in
the general $sl(3)$ symmetric model considered in
\cite{resh-su3}.
This is due to the fact that we considered directly the form factors
and not the scalar products, and that we
used the explicit
form of the coordinate wave functions, as opposed to the 
algebraic construction of \cite{resh-su3}.
Therefore we do not find
singularities associated with coinciding magnonic rapidities,
 but the
kinematical poles of the particle rapidities yield also the
``inhomogeneous'' terms. 

The magnonic form factors have the structure
\begin{equation}
\label{structure}
  \begin{split}
&     \mathcal{F}_N\big(\{p\}_{N},\{\nu\}_{M'}|\{k\}_{N},\{\mu\}_{M}\big)=
c^{M'+M} \mathcal{P}_N\big(\{p\}_{N},\{\nu\}_{M'}|\{k\}_{N},\{\mu\}_{M}\big)\times\\
&\prod_{j>l} F_{\text{min}}(k_j-k_l) 
\prod_{j>l} F_{\text{min}}(p_l-p_j) 
\prod_{j,l} \frac{1}{k_j-p_l} 
\prod_{j,l} \frac{1}{\mu_j-k_l+ic/2}
\prod_{j,l} \frac{1}{\nu_j-p_l-ic/2}.
  \end{split}
\end{equation}
Here $F_{\text{min}}(k)$ is the so-called minimal two-particle form
factor responsible for the exchange properties; it satisfies the relation
\begin{equation*}
  F_{\text{min}}(u)=S_1(-u)F_{\text{min}}(-u).
\end{equation*}
The solutions are given by
\begin{equation}
  \label{Fmin}
  F_{\text{min}}(u)=
  \begin{cases}
    u & \text{for fermions}\\
\frac{u}{u-ic}  &  \text{for bosons.}
  \end{cases}
\end{equation}
$\mathcal{P}_N$ is a polynomial which is symmetric with respect
to all four sets of variables. 

\begin{thmlem}
  \label{PNdegree}
The maximal degree of $\mathcal{P}_N$ in its
variables depends on the operator 
in question and the statistics of the model and is given as follows:
\begin{itemize}
\item In the fermionic case:
  \begin{itemize}
  \item Field operators: 
 $\mathcal{P}_N$ is of order $M'$ in
    the $p$ variables and of order $M-1$ in the $k$ variables.
\item Bilinear operators: 
$\mathcal{P}_N$ is of order $M'+1$ in
    the $p$ variables and of order $M+1$ in the $k$ variables.
  \end{itemize}
\item In the bosonic case:
  \begin{itemize}
  \item Field operators: $\mathcal{P}_N$ is of order $N-1+M'$ in
    the $p$ variables and of order $N-1+M$ in the $k$ variables.
\item Bilinear operators: 
$\mathcal{P}_N$ is of order $N+M'$ in
    the $p$ variables and of order $N+M$ in the $k$ variables.
  \end{itemize}
\end{itemize}
\end{thmlem}
\begin{proof}
The total degree of the form
factor in the particle rapidities can be established using the
arguments given in Lemmas \ref{field-asy} and \ref{density-asy} and
the explicit form of the wave function \eqref{hukkle}.
Then the degree of $\mathcal{P}_N$ can be read off from
\eqref{structure}. The main difference between the degrees in the
fermionic and bosonic cases is a result of the different structure of
the minimal two-particle form factor.
\end{proof}

It follows from the Lemma that in the fermionic case \eqref{psybient} is restrictive
enough to fix the form factors completely. 

On the other hand, in the
bosonic case the degree of $\mathcal{P}_N$  is typically higher than the number
of conditions provided by \eqref{psybient}. 
Further constraints can be found by sending one of the
magnonic rapidities to infinity:  Lemma \eqref{we} and the
Wigner-Eckhart theorem provide additional relations between different
form factors. 
However, at the present moment it is not clear whether \eqref{psybient} can be
supplemented with other conditions which would fix the form
factors. These questions are left for further research.

\section{A few explicit solutions for the Form Factors}

\label{solutions}

In this section we present solutions to the recursive equations
\eqref{psybient}. We only consider cases when there is at most one
$\ket{-}$ spin in the two states; for certain operators this leads to homogeneous recursive
equations which can be solved using generalizations of already known techniques. The
solution of the full inhomogeneous equations would require new
techniques and is outside the scope of the present paper.

As an independent check of our results we also evaluated the
coordinate Bethe Ansatz expressions for a low number of particles
using \texttt{Mathematica}.
In the
cases $N=1,2,3$ we performed the comparisons analytically, whereas in the 
cases $N=4,5$ we could only do numeric checks. In all cases we found complete
agreement. This provides a strong justification for the results
presented below, especially for the bilinear operators in
the bosonic case, where the recursion relations 
do not fix the form factors completely.

\subsection{Fermions}

In the fermionic case we have $\tilde S_{1}=1$ leading to simple
recursive equations. Form factors with only $\ket{+}$ spins (no
magnonic rapidities) vanish identically, because the one-component
fermionic model is a free theory (on the level of form factors this
follows from $\tilde S_{1}=1$ leading to vanishing kinematic poles). 

\subsubsection{The matrix elements  $\bra{N,0}\Psi_-\ket{N+1,1}$}

Here we consider matrix elements of the $\ket{-}$ spin field operator
between a completely polarized state and a state with only one $\ket{-}$
spin:
\begin{equation*}
  \mathcal{F}^-_N(p_1,\dots,p_N|k_0,\dots,k_N,\mu).
\end{equation*}

It follows from \eqref{psybient} that the
pole at $p_N=k_{N}$ is given by
\begin{equation}
\label{ss1}
\begin{split}
&\mathcal{F}^-_N(p_1,\dots,p_N|k_0,\dots,k_N,\mu)\sim
\frac{i}{k_N-p_N}\times\\
&\hspace{2.5cm}\left(
1-\frac{\mu-k_N-ic/2}{\mu-k_N+ic/2}\right)
\mathcal{F}^-_{N-1}(p_1,\dots,p_{N-1}|k_0,\dots,k_{N-1},\mu)\\
&\hspace{2cm}=\frac{c}{(k_N-p_N)(k_N-\mu-ic/2)} \mathcal{F}^-_{N-1}(p_1,\dots,p_{N-1}|k_0,\dots,k_{N-1},\mu).
\end{split}
\end{equation}
The starting point for the recursion is the formal value at $N=0$:
\begin{equation*}
  \mathcal{F}^-_0(k_0,\mu)=\frac{-ic}{k_0-\mu-ic/2}.
\end{equation*}
The solution to equation \eqref{ss1} is
\begin{equation*}
  \mathcal{F}^-_N(p_1,\dots,p_N|k_1,\dots,k_{N+1},\mu)=
-i \left(\prod_j \frac{c}{k_j-\mu-ic/2}\right)
\frac{\prod_{i<j} k_{ji}\prod_{i>j} p_{ij}}{\prod_{j,l} (k_j-p_l)}.
\end{equation*}

\subsubsection{The matrix elements  $\bra{N,1}\Psi_+\ket{N+1,1}$}

If $N>1$ then the ``inhomogeneous'' term vanishes and we obtain the
recursion equations
\begin{equation}
\label{ss2}
\begin{split}
&  \mathcal{F}^+_{N}(p_1,\dots,p_N,\nu|k_0,\dots,k_N,\mu)\sim 
\frac{i}{k_N-p_N} \\
&\left(
1-\frac{\mu-k_N-ic/2}{\mu-k_N+ic/2}\frac{\nu-p_N+ic/2}{\nu-p_N-ic/2}\right)
\mathcal{F}^+_{N-1}(p_1,\dots,p_{N-1},\nu|k_0,\dots,k_{N-1},\mu)=\\
&\frac{c(\mu-\nu)}{(k_N-p_N)(\mu-k_N+ic/2)(\nu-p_N-ic/2)}
\mathcal{F}^+_{N-1}(p_1,\dots,p_{N-1},\nu|k_0,\dots,k_{N-1},\mu).
\end{split}
\end{equation}
At $N=1$ the inhomogeneous
term in \eqref{psybient} is the only contribution yielding for the pole
at $k_1\to p_1$:
\begin{equation*}
\begin{split}
  \mathcal{F}^+_1(p_1,\nu|k_0,k_1,\mu)&\sim
\frac{i}{k_1-p_1} 
\frac{ic}{\mu-k_1+ic/2}\frac{-ic}{\nu-p_1-ic/2} [S_{\frac{1}{2}}(\mu-k_0)-1]
\mathcal{F}^+_0(k_0)\\
&=\frac{i}{k_1-p_1} 
\frac{ic}{\mu-k_1+ic/2}\frac{-ic}{\nu-p_1-ic/2}\frac{-ic}{\mu-k_0+ic/2}.
\end{split}
\end{equation*}
There is an analogous relation for the pole at $p_1\to k_0$. The
solution is
\begin{equation*}
  \mathcal{F}^+_1(p_1,\nu|k_0,k_1,\mu)=
-i \frac{-ic}{\nu-p_1-ic/2}
\frac{ic}{\mu-k_1+ic/2}\frac{ic}{\mu-k_0+ic/2}
\frac{k_0-k_1}{(k_1-p_1)(k_0-p_1)}.
\end{equation*}

The solution to the recursive equation then reads
\begin{equation*}
\begin{split}
&  \mathcal{F}^+_{N}(p_1,\dots,p_N,\nu|k_0,\dots,k_N,\mu)=\\
&c^3 \big(c(\mu-\nu)\big)^{N-1}
\left(\prod_j \frac{1}{\mu-k_j+ic/2}\right)
\left(\prod_j \frac{1}{\nu-p_j-ic/2}\right)
\frac{\prod_{i<j} k_{ji}\prod_{i>j} p_{ij}}{\prod_{j,l} (k_j-p_l)}.
\end{split}
\end{equation*}

\subsubsection{The matrix elements $\bra{N,0}\Psi_+^\dagger \Psi_-\ket{N,1}$}

The recursive relation is analogous to \eqref{ss1}:
\begin{equation}
\label{ss3}
\begin{split}
&\mathcal{F}^{+-}_N(p_1,\dots,p_N|k_1,\dots,k_N,\mu)\sim \frac{i}{k_N-p_N}\times\\
&\hspace{1cm}\frac{c}{(k_N-p_N)(k_N-\mu-ic/2)} \mathcal{F}^{+-}_{N-1}(p_1,\dots,p_{N-1}|k_1,\dots,k_{N-1},\mu).
\end{split}
\end{equation}
The starting value is 
\begin{equation*}
  \mathcal{F}^{+-}_1(p_1|k_1,\mu)=\frac{ic}{\mu-k_1+ic/2}.
\end{equation*}
The solution is
\begin{equation}
  \label{freschbach}
   \mathcal{F}^{+-}(p_1,\dots,p_N|k_1,\dots,k_{N},\mu)=-ic^N
\left(\sum_j (k_j-p_j)\right) \prod_j \frac{1}{\mu-k_j+ic/2}\det \frac{1}{k_j-p_l}.
\end{equation}

\subsubsection{The matrix elements $\bra{N,1}\Psi_-^\dagger \Psi_-\ket{N,1}$}

The recursion equation reads
\begin{equation}
\label{ss4}
\begin{split}
&  \mathcal{F}^{--}_{N}(p_1,\dots,p_N,\nu|k_1,\dots,k_N,\mu)\sim \frac{i}{k_N-p_N}\times\\
&\frac{c(\mu-\nu)}{(k_N-p_N)(\mu-k_N+ic/2)(\nu-p_N-ic/2)}
\mathcal{F}^{--}_{N-1}(p_1,\dots,p_{N-1},\nu|k_1,\dots,k_{N-1},\mu).
\end{split}
\end{equation}
The starting value is
\begin{equation*}
  \mathcal{F}^{--}_1(p_1,\nu|k_1,\mu)=\frac{c^2}{(\mu-k_N+ic/2)(\nu-p_N-ic/2)}.
\end{equation*}
The solution is
\begin{equation}
  \label{frechbax}
\begin{split}
   \mathcal{F}^{--}_N(p_1,\dots,p_N,\nu|k_1,\dots,k_{N},\mu)= \big(c(\mu-\nu)\big)^{N-1}
\left(\sum_j (k_j-p_j)\right)\\
\times  \frac{c^2}{\prod_j (\mu-k_j+ic/2)(\nu-p_j-ic/2)}
 \det \frac{1}{k_j-p_l}.
\end{split}
\end{equation}

\subsubsection{The matrix elements $\bra{N,1}\Psi_+^\dagger \Psi_+\ket{N,1}$}

Here the recursion relation is similar to the previous case:
\begin{equation}
\label{ss5}
\begin{split}
&  \mathcal{F}^{++}_{N}(p_1,\dots,p_N,\nu|k_1,\dots,k_N,\mu)\sim \\
&\frac{c(\mu-\nu)}{(k_N-p_N)(\mu-k_N+ic/2)(\nu-p_N-ic/2)}
\mathcal{F}^{++}_{N-1}(p_1,\dots,p_{N-1},\nu|k_1,\dots,k_{N-1},\mu).
\end{split}
\end{equation}
However, this relation is valid only for $N>2$.  At $N=2$ only the
inhomogeneous term contributes and we
have for example the pole at $p_2\to k_2$:
\begin{equation*}
\begin{split}
&  \mathcal{F}^{++}_2(p_1,p_2,\nu|k_1,k_2,\mu)\sim\\
&\hspace{1cm}\sim\frac{i}{k_2-p_2} 
\frac{ic}{\mu-k_2+ic/2}\frac{-ic}{\nu-p_2-ic/2} [S_{\frac{1}{2}}(\mu-k_1)S_{\frac{1}{2}}(p_1-\nu)-1]
\mathcal{F}^{++}_1(p_1|k_1)\\
&\hspace{1cm}=\frac{i}{k_2-p_2} 
\frac{ic}{\mu-k_2+ic/2}\frac{-ic}{\nu-p_2-ic/2}
\frac{ic(\mu-\nu+p_1-k_1)}{(\mu-k_1+ic/2)(\nu-p_1-ic/2)}.
\end{split}
\end{equation*}
There are similar pole relations for the other residues. The solution
is
\begin{equation*}
\begin{split}
&  \mathcal{F}^{++}_2(p_1,p_2,\nu|k_1,k_2,\mu)=\\
&\hspace{0.5cm}c^3 \frac{(k_1+k_2-p_1-p_2)(\nu-\mu+k_1+k_2-p_1-p_2)}
{(\mu-k_1+ic/2)(\mu-k_2+ic/2)(\nu-p_1-ic/2)(\nu-p_2-ic/2)}\det
\frac{1}{k_j-p_l}.
\end{split}
\end{equation*}
The solution to the recursive equations is then
\begin{equation}
\begin{split}
&  \mathcal{F}^{++}_N (p_1,\dots,p_N,\nu|k_1,\dots,k_N,\mu)=
c (c(\mu-\nu))^{N-2}  \frac{c^2}{\prod_j(\mu-k_j+ic/2)(\nu-p_j-ic/2)} \\
&\hspace{2cm}\times \big(\sum_j (k_j-p_j)\big)\big(\nu-\mu+\sum_j (k_j-p_j)\big)
\det \frac{1}{k_j-p_l}.
\end{split}
\end{equation}

\subsection{Bosons}

In the bosonic case $\tilde S_1(u)=(u-ic)/(u+ic)$ and the recursive
equations for the polarized states (no magnonic rapidities) 
produce the form factors of the Lieb-Liniger model. Therefore it is
expected, that in those cases, where the recursive equations are
homogeneous with fixed magnonic rapidities, the form factors can be
obtained as generalizations of the Lieb-Liniger formulas. 

There are
three cases where this occurs, with one magnon at most in each of the
states:
\begin{itemize}
\item Matrix elements of the down spin field operator:
  $\bra{N,0}\Psi_-\ket{N+1,1}$
\item Matrix elements of the spin-flip operator:
 $\bra{N,0}\Psi_+^\dagger\Psi_-\ket{N,1}$
\item Matrix elements of the density of down spins:
 $\bra{N,1}\Psi_-^\dagger\Psi_-\ket{N,1}$
\end{itemize}

The next matrix elements to consider would be the ones
$\bra{N,1}\Psi_+^\dagger\Psi_+\ket{N,1}$. However, in this case there
is an inhomogeneous term 
at each step of the recursion; it is given by a Lieb-Liniger density form
factor. Therefore, already this case is outside the scope of the
present methods and it requires new ideas.

\subsubsection{The matrix elements $\bra{N,0}\Psi_-\ket{N+1,1}$}

The residue property at $p_N\to k_{N+1}$ is
\begin{equation}
  \label{mirage}
\begin{split}
     \mathcal{F}^-_N(p_1,\dots,p_N|k_1,\dots,k_{N+1},\mu)\sim
\frac{i}{k_{N+1}-p_N}\times\hspace{6cm}\\
\times\left[1-\frac{\mu-k_{N+1}+ic/2}{\mu-k_{N+1}-ic/2}
\prod_{j=1}^N S(k_{j,N+1})\prod_{k=1}^{N-1}S(p_{Nk})
\right] 
\mathcal{F}^-_{N-1}(p_1,\dots,p_{N-1}|k_1,\dots,k_{N},\mu).
\end{split}
\end{equation}

The starting point for the recursion is the formal result
\begin{equation*}
  \mathcal{F}^-_0(k_0,\mu)=\frac{-ic}{\mu-k_0-ic/2}.
\end{equation*}

It follows from Lemma \ref{PNdegree} that the recursion with respect
to the $p$ variables completely fixes this form factor. 
We find the following solution:
\begin{equation}
\label{hawker}
   \mathcal{F}_N^-(p_1,\dots,p_N|k_1,\dots,k_{N+1},\mu)=
\prod_{i>j} \frac{k_i-k_j+ic}{p_i-p_j+ic}
 \frac{-ic}{\prod_j(\mu-k_j-ic/2)}\ \det \mathbb{M},
\end{equation}
where $\mathbb{M}$ is an $N\times N$ matrix with entries
\begin{equation*}
  \mathbb{M}_{jk}=M_{jk}-M_{N+1,k}
\end{equation*}
with
\begin{equation}
\label{ready-or-not}
  M_{jk}=t(p_k,k_j)
h_{\frac{1}{2}}(\mu,k_j)
\frac{\prod_{m=1}^N h(p_m,k_j)}{\prod_{m=1}^{N+1}
    h(k_m,k_j)}+ t(k_j,p_k)h_{\frac{1}{2}}(k_j,\mu) \frac{\prod_{m=1}^N h(k_j,p_m)}{\prod_{m=1}^{N+1}
    h(k_j,k_m)}
\end{equation}
and
\begin{equation*}
  h(u)={u+ic} \qquad\qquad h_{\frac{1}{2}}(u)=u+ic/2 \qquad\qquad t(u)=\frac{-c}{u(u+ic)}.
\end{equation*}

The first example is
\begin{equation*}
  \mathcal{F}^-_1(p|k_0,k_1,\mu)=
\frac{k_1-k_0}{k_1-k_0-ic}
\frac{ic^2(k_0+k_1-2\mu)}{(\mu-k_0-ic/2)(\mu-k_1-ic/2)(p-k_0)(p-k_1)}.
\end{equation*}

\bigskip

An alternative expression for the same function is a
generalization of the formula \eqref{finally}:
\begin{equation}
\label{maraeleglesz}
\begin{split}
&   \mathcal{F}^{-}_N(\{p\}_N|\{k\}_{N+1},\mu)=\\
&\hspace{1cm}\frac{i P}{c^{N+1}}
\prod_{j>l}\frac{1}{k_j-k_l-ic} \prod_{j>l}\frac{1}{p_j-p_l+ic}
  \frac{-ic}{\prod_j (\mu-k_j-ic/2)}
\prod_{j,k} \frac{1}{k_j-p_l}.
\end{split}
\end{equation}
Here
\begin{equation}
\begin{split}
 P=&\mathop{\sum_{\alpha_j=0,1}}
\sum_{\beta_l=0,1}  (-1)^{\sum_j \alpha_j+\sum_l \beta_l }
\prod_{i1<i2} (p_{i1}-p_{i_2}+(\alpha_{i_1}-\alpha_{i_2})ic)\times \\
&\times \prod_{i1<i2} (k_{i1}-k_{i_2}+(\beta_{i_1}-\beta_{i_2})ic)
\prod_{i1,i2} (p_{i1}-k_{i_2}-(\alpha_{i_1}-\beta_{i_2})ic)  \times\\
&\times\prod_j (\mu-k_j+(2\beta_j-1)ic/2).
\end{split}
\end{equation}
It is easy to see that \eqref{maraeleglesz} satisfies all the
conditions and the starting value for the recursion, therefore it is
equivalent to \eqref{hawker}.

\subsubsection{The matrix elements $\bra{N,0}\Psi^\dagger_+\Psi_-\ket{N,1}$}

In this case the recursion relation is analogous to \eqref{mirage}
and we found the following solution, which is given by a generalization of \eqref{FFK}:
\begin{equation}
  \label{FFKMSF}
  \begin{split}
& 
 \mathcal{F}_N^{+-}(\{p\}|\{k\},\mu)=\frac{i}{c} (-1)^{N(N-1)/2}
\prod_{j>l}\frac{1}{k_j-k_l-ic} \prod_{j>l}\frac{1}{p_j-p_l+ic}\\
&\hspace{3cm}\times \prod_{o=1}^N \prod_{l=1}^N (k_o-p_l+ic)\times 
\det V \times \frac{-ic}{\prod_j (\mu-k_j-ic/2)}.
  \end{split}
\end{equation}
Here $V$ is an $(N+1)\times (N+1)$ matrix with entries
\begin{equation}
  \begin{split}
 V_{jl}&=(p_l-\mu+ic/2) \tilde t(k_j,p_l)+\\ 
&\hspace{0.5cm}+(p_l-\mu-ic/2)\tilde t(p_l,k_j)
\prod_{o=1}^N
\frac{(p_l-k_o+ic)(p_l-p_o-ic)}{(p_l-k_o-ic)(p_l-p_o+ic)}, \qquad  j,l=1\dots N\\
V_{N+1,j}&=\prod_{o=1}^N \frac{p_o-p_j+ic}{k_o-p_j+ic} \qquad\text{and}\quad
V_{j,N+1}=1,\qquad  j=1\dots N\\
V_{N+1,N+1}&=0
  \end{split}
\end{equation}
and
\begin{equation*}
\tilde t(u)=\frac{-i}{u(u+ic)}  
\end{equation*}

An alternative representation for the same function is given by
\begin{equation}
\label{korsakoff23}
\begin{split}
&   \mathcal{F}^{+-}_N(\{p\}_N|\{k\}_{N},\mu)=\\
&\hspace{1cm}\frac{-i}{2c^{N}}
\prod_{j>l}\frac{1}{k_j-k_l-ic} \prod_{j>l}\frac{1}{p_j-p_l+ic}
\times  P\times \frac{-ic}{\prod_j (\mu-k_j-ic/2)}
\prod_{j,k} \frac{1}{k_j-p_l}
\end{split}
\end{equation}
with
\begin{equation}
\begin{split}
 P=&\mathop{\sum_{\alpha_j=0,1}}
\sum_{\beta_l=0,1}  (-1)^{\sum_j \alpha_j+\sum_l \beta_l }
\prod_{i1<i2} (p_{i1}-p_{i_2}+(\alpha_{i_1}-\alpha_{i_2})ic)\times \\
&\times \prod_{i1<i2} (k_{i1}-k_{i_2}+(\beta_{i_1}-\beta_{i_2})ic)
\prod_{i1,i2} (p_{i1}-k_{i_2}-(\alpha_{i_1}-\beta_{i_2})ic)  \times\\
&\times\prod_j (\mu-k_j+(2\beta_j-1)ic/2)\times
 \sum_{j=1}^N \left((-1)^{\alpha_j}+ (-1)^{\beta_j}\right)
\end{split}
\end{equation}

The first two cases are given explicitly as
\begin{equation*}
     \mathcal{F}^{+-}_1(p|k,\mu)=\frac{-ic}{\mu-k-ic/2} 
\end{equation*}
\begin{equation*}
\begin{split}
 \mathcal{F}^{+-}_2(p_1,p_2|k_1,k_2,\mu)=&
 (c^2 - 4 k_1 k_2 + 2 (k_1+k_2-p_1-p_2) \mu +  (k_1+k_2)
 (p_1+p_2))\times \\
&\frac{-ic}{(\mu-k_1-ic/2)(\mu-k_2-ic/2)}
\frac{k_2-k_1}{k_2-k_1-ic} \frac{p_2-p_1}{p_2-p_1+ic}\times
\\
&\frac{-c (k_1+k_2-p_1-p_2)}{(k_1-p_1)(k_1-p_2)(k_2-p_1)(k_2-p_2)}.
  \end{split}
\end{equation*}

In the present case the kinematic recursion relation itself is not
sufficient to fix the form factor. We compared the two representations
above to each other and to the coordinate Bethe Ansatz results. We
performed the analytical check up to $N=3$ and numerical checks for
$N=4,5$ with \texttt{Mathematica}. Complete agreement was found, which
provides a very strong justification for the general case of $N>5$.

\subsubsection{The matrix elements $\bra{N,1}\Psi_-^\dagger \Psi_-\ket{N,1}$}

The residue property at $p_N\to k_{N}$ is
\begin{equation}
  \label{mirage2}
\begin{split}
 &    \mathcal{F}^{--}_N(p_1,\dots,p_N,\nu|k_1,\dots,k_{N},\mu)\sim
\frac{i}{k_{N}-p_N}\times\\
&\hspace{2cm}\left[1-\frac{\mu-k_{N}+ic/2}{\mu-k_{N}-ic/2}\frac{\nu-p_{N}-ic/2}{\nu-p_{N}+ic/2}
\prod_{j=1}^{N-1} S(k_{jN})\prod_{k=1}^{N-1}S(p_{Nk})
\right] \times\\
&\hspace{3cm}\mathcal{F}^{--}_{N-1}(p_1,\dots,p_{N-1},\nu|k_1,\dots,k_{N-1},\mu).
\end{split}
\end{equation}

One solution is
\begin{equation}
  \label{FFKM}
  \begin{split}
& \mathcal{F}_N^{--}(\{p\},\nu|\{k\},\mu)=\frac{-i}{c} (-1)^{N(N+1)/2}
\prod_{j>l}\frac{1}{k_j-k_l-ic} \prod_{j>l}\frac{1}{p_j-p_l+ic}
\\ & \hspace{2cm}\times \det V\times \prod_{o=1}^N \prod_{l=1}^N (k_o-p_l+ic)
\times  \frac{c^2}{\prod_j (\mu-k_j-ic/2)(\nu-p_j+ic/2)}.
  \end{split}
\end{equation}
Here $V$ is an $(N+1)\times (N+1)$ matrix with entries
\begin{equation}
  \begin{split}
 V_{jl}&=(p_l-\mu+ic/2)(p_l-\nu-ic/2) \tilde t(k_j,p_l)+\\ &
(p_l-\mu-ic/2)(p_l-\nu+ic/2)\tilde t(p_l,k_j)
\prod_{o=1}^N
\frac{(p_l-k_o+ic)(p_l-p_o-ic)}{(p_l-k_o-ic)(p_l-p_o+ic)}, \qquad j,l=1\dots N\\
V_{N+1,j}&=\prod_{o=1}^N \frac{p_o-p_j+ic}{k_o-p_j+ic} \qquad\text{and}\quad
V_{j,N+1}=1,\qquad   j=1\dots N\\
V_{N+1,N+1}&=0
  \end{split}
\end{equation}
and
\begin{equation*}
\tilde t(u)=\frac{-i}{u(u+ic)}  .
\end{equation*}

An alternative expression reads
\begin{equation}
\label{korsakoff2}
\begin{split}
&   \mathcal{F}^{--}_N(\{p\}_N,\nu|\{k\}_{N},\mu)=\\
&\hspace{1cm}c^{-(N+1)}\frac{\sum_j (k_j-p_j)}{\nu-\mu+2\sum_j (k_j-p_j)}
\prod_{j>l}\frac{1}{k_j-k_l-ic} \prod_{j>l}\frac{1}{p_j-p_l+ic}\\
&\hspace{1cm}\times  P\times \frac{c^2}{\prod_j(\mu-k_j-ic/2)(\nu-p_j+ic/2)} 
\prod_{j,k} \frac{1}{k_j-p_l}.
\end{split}
\end{equation}
Here
\begin{equation}
\begin{split}
 P=&\mathop{\sum_{\alpha_j=0,1}}
\sum_{\beta_l=0,1}  (-1)^{\sum_j \alpha_j+\sum_l \beta_l }
\prod_{i1<i2} (p_{i1}-p_{i_2}+(\alpha_{i_1}-\alpha_{i_2})ic)\times \\
&\times \prod_{i1<i2} (k_{i1}-k_{i_2}+(\beta_{i_1}-\beta_{i_2})ic)
\prod_{i1,i2} (p_{i1}-k_{i_2}-(\alpha_{i_1}-\beta_{i_2})ic)  \times\\
&\times\prod_j (k_j-\mu-(2\beta_j-1)\frac{ic}{2})
(p_j-\nu-(2\alpha_j-1)\frac{ic}{2})
\end{split}
\end{equation}

The first two cases are given explicitly as
\begin{equation*}
     \mathcal{F}^{--}_1(p,\nu|k,\mu)=\frac{c^2}{(\mu-k-ic/2)(\nu-p+ic/2)} 
\end{equation*}

\begin{equation}
\label{twentythree}
\begin{split}
& \mathcal{F}^{--}_2(p_1,p_2,\nu|k_1,k_2,\mu)=
\frac{c (k_1+k_2-p_1-p_2)
  (k_1-k_2)(p_1-p_2)}{(k_1-p_1)(k_1-p_2)(k_2-p_1)(k_2-p_2)}\times\\
&\hspace{2cm}\frac{c^2}{(\mu-k_1-ic/2)(\mu-k_2-ic/2)(\nu-p_1+ic/2)(\nu-p_2+ic/2)}\times
\\
&\hspace{2cm}(c^2 (\tilde \mu - \tilde \nu) - (k_1-k_2)^2 \tilde \nu 
-   2 (k_1+k_2-p_1-p_2) \tilde \mu \tilde \nu  + 
 + \tilde \mu (p_1-p_2)^2 ).
  \end{split}
\end{equation}
Here we used the auxiliary variables
\begin{equation*}
  \tilde \mu=\mu-\frac{k_1+k_2}{2}\qquad\qquad
\tilde \nu=\nu-\frac{p_1+p_2}{2}.
\end{equation*}

Once again we compared the two representations
above to each other and to the coordinate Bethe Ansatz results up to
$N=5$ and found complete agreement. 

\section{The nested Bethe Ansatz in a finite volume}

\label{fftcsa1hehe}

The previous sections were concerned with the infinite volume form
factors of the multi-component systems. Here we consider the
two-component case in a finite volume. 

First we recall the conditions for the periodicity of the nested Bethe
Ansatz wave function.

\begin{thm}
  The wave function \eqref{hukkle} is periodic in a finite
  volume $L$ iff the rapidities $\{p\}_N$ and $\{\mu\}_M$ satisfy the following coupled set of
  equations:
\begin{equation}
\label{nestedBA0}
  e^{ip_j L}
=
\mathop{\prod_{l=1}^N}_{l\ne j} 
\tilde S_1(p_l-p_j)
\prod_{l=1}^M
\frac{\mu_l-p_j+i\sigma c/2}{\mu_l-p_j-i\sigma c/2}
\qquad j=1\dots N
\end{equation}
\begin{equation}
\label{nestedBA1}
  \prod_{j=1}^N
  \frac{\mu_l-p_j+ic/2}{\mu_l-p_j-ic/2}=
\mathop{\prod_{j=1}^M}_{j\ne l}
\frac{\mu_l-\mu_j+ic}{\mu_l-\mu_j-ic}
\qquad l=1\dots M.
\end{equation}
\end{thm}
\begin{proof}
For simplicity we only consider the periodicity with respect to the
variable $x_1$.  By symmetry it is enough to consider the periodicity
condition
\begin{equation*}
  \chi_N(-L/2,x_2,\dots,x_{N}|\{p\}_N,\{\mu\}_M)= \chi_N(x_2,\dots,x_{N},L/2|\{p\}_N,\{\mu\}_M)
\end{equation*}
in the domain
\begin{equation*}
  -L/2 < x_2 < \dots < x_{N} < L/2.
\end{equation*}
For simplicity we only consider the coefficient of the term
proportional to
\begin{equation}
\label{exp123}
  e^{i(x_1p_1+\dots+x_Np_N)},
\end{equation} 
this will lead to the condition \eqref{nestedBA0} with $j=1$. 
At $x_1=-L/2$ the coefficient is equal to $\omega_N(\{p\},\{\mu\})$. On the other
hand, at $x_1=L/2$ the coefficient is 
\begin{equation*}
\rho(Q^{-1})  \mathcal{Q}(Q,p) \omega_N(\{p\},\{\mu\}),
\end{equation*}
where $Q$ is the permutation giving $Qp=\{p_2,\dots,p_{N},p_1\}$. 
Thus we obtain the condition
\begin{equation}
\label{per1}
  \omega_N(\{p\},\{\mu\})=
e^{ip_1L}
\rho(Q^{-1})  \mathcal{Q}(Q,p) \omega_N(\{p\},\{\mu\}).
\end{equation}

It follows from the definition \eqref{T1} and  from 
$X_{a,1}(0)=\sigma P_{a,1}$ that
\begin{equation}
  T(p_1|\{p\})=\sigma P_{12}P_{23}\dots P_{N-1,N}P_{a,N} 
\hat S_{1,N} \hat S_{1,N-1} \dots  \hat S_{1,2} .
\end{equation}
Here $\hat S_{1,j}=Y^{j-1,j}_{1,j}$. Taking the trace with respect to
the auxiliary space
\begin{equation}
\label{per2}
  t(p_1|\{p\})=\sigma^{N}\rho(Q^{-1}) \hat S_{1,N} \hat S_{1,N-1} \dots
  \hat S_{1,2}=\sigma^{N} \rho(Q^{-1})\mathcal{Q}(Q,p).
\end{equation}
Using \eqref{per1} and \eqref{per2} the periodicity condition is
expressed as
\begin{equation}
   \omega_N(\{p\},\{\mu\})=
e^{ip_1L} \sigma^N t(p_1|\{p\})  \omega_N(\{p\},\{\mu\}).
\end{equation}

The eigenvalue equation
\begin{equation}
\label{second}
   t(u|\{p\})\omega_N(\{p\},\{\mu\}) =  \Lambda(u|\{p\}) \omega_N(\{p\},\{\mu\})
\end{equation}
can be solved by the standard methods of algebraic Bethe Ansatz.
It is known \cite{korepinBook} that the above equation is satisfied
whenever the magnonic rapidities are solutions to the inhomogeneous
Bethe equation \eqref{nestedBA1}.
Then the  eigenvalues read
\begin{equation}
\label{eig2}
  \Lambda(u|\{p\})=
\prod_{j=1}^N S_1(u-p_j)\left(
\prod_{l=1}^M \frac{\mu_l-u-i\sigma c/2}{\mu_l-u+i\sigma c/2}
+
\prod_{j=1}^N 
\frac{u-p_j}{u-p_j-i\sigma c}
\prod_{l=1}^M \frac{\mu_l-u+3i\sigma c/2}{\mu_l-u+i\sigma c/2}
\right).
\end{equation}
Substituting $u=p_1$ results in
\begin{equation}
\Lambda(p_1|\{p\})=  \prod_{j=1}^N \tilde S_1(u-p_j)
\prod_{l=1}^M \frac{\mu_l-u-i\sigma c/2}{\mu_l-u+i\sigma c/2}.
\end{equation}
with $\tilde S_1(u)=\sigma S_1(u)$. This leads to \eqref{nestedBA0}
with $j=1$. The other conditions with $j=2,\dots,N$ follow from the
symmetry properties of the wave function. It can be checked using the
Yang-Baxter equation that these equations guarantee the periodicity
for the coefficients of all exponentials and not only \eqref{exp123}
considered here.
\end{proof}

It is important to note that the states obtained by the finite
solutions to the equations
\eqref{nestedBA0}-\eqref{nestedBA1} do not span the full Hilbert space:
 it is known that the
Algebraic Bethe Ansatz construction \eqref{omegadef} only produces the states
which are highest weight with respect to the overall $SU(2)$ symmetry
\cite{Faddeev-ABA-intro}. In order to obtain the non-highest weight
states one must act with spin lowering operator. This is known to be
equivalent (in the proper normalization) to sending one magnonic
rapidity to infinity. Supplied with these solutions the eqs.
\eqref{nestedBA0}-\eqref{nestedBA1} are believed to yield a complete
set of states.

\subsection{Form factors in finite volume}

We define the finite volume form factors in the same way as in the
case of the Lieb-Liniger model. For example, in the case of the field
operator the form factor is given by
\begin{equation}
\begin{split}
&  \mathbb{F}^l_N(\{p\}_N,\{\nu\}_{M'},\{k\}_{N+1},\{\mu\}_M)=
\sqrt{N+1} \int_{-L/2}^{L/2}  dx_1 \dots dx_N\
\\
&\hspace{2cm}\times \Big\langle  \chi_{N}(x_1,\dots,x_N|\{p\},\{\nu\})
\Big| U^{(0)}_l \chi_{N+1}(0,x_1,\dots,x_N|\{k\},\{\mu\})
\Big\rangle_N.
\end{split}
\end{equation}
Spin conservation requires $M'=M+\frac{l-1}{2}$. 
Analogous definitions can be given for the form factors of the
bilinear operators.

\begin{thm}
\label{ugyanaz}
The form factors are the same in finite and infinite volume. In other
words, if both sets $\{\{p\},\{\nu\}\}$ and $\{\{k\},\{\mu\}\}$ are
solution to the nested Bethe
equations and there are no coinciding particle rapidities ($p_j\ne k_l$), then 
\begin{equation}
\label{sudaya1}
   \mathbb{F}^l_N(\{p\}_N,\{\nu\}_{M'},\{k\}_{N+1},\{\mu\}_M)=
 \mathcal{F}^l_N(\{p\}_N,\{\nu\}_{M'},\{k\}_{N+1},\{\mu\}_M).
\end{equation}
An analogous relation holds for the
matrix elements of the bilinear operators.
\end{thm}
\begin{proof}
The theorem can be proven with the same arguments as Theorem
\ref{fftcsa}. The key idea is that when the wave functions are
periodic, the contributions of the Newton-Leibniz formula at $x=\pm
L/2$ cancel each other and  only the contributions at $x=\pm 0$
remain, which exactly coincide with those given by the infinite volume
regularization of the oscillating integrals.
\end{proof}

Finally we note that the normalized form factors are given by
\begin{equation}
  \label{full-answer}
  \begin{split}
    \bra{\{p\}_N,\{\nu\}_{M'}}\Psi_l\ket{\{k\}_{N+1},\{\mu\}_{M}}=
\frac{ \mathcal{F}^l_N(\{p\}_N,\{\nu\}_{M'},\{k\}_{N+1},\{\mu\}_{M})}
{\sqrt{\mathcal{N}(\{p\}_N,\{\nu\}_{M'})\ \mathcal{N}(\{k\}_{N+1},\{\mu\}_M)}},
  \end{split}
\end{equation}
and similarly for the bilinear operators. Here $\mathcal{N}$ denotes
the norm of the eigenstates
and it reads \cite{resh-su3,pang-zhao-norms,Hubbard-norms}
\begin{equation}
\label{norms}
  \mathcal{N}(\{p\}_N,\{\mu\}_M)=
c^M\ \det \mathcal{G} \prod_{1\le j<k \le M} \left(1+\frac{c^2}{(\mu_j-\mu_k)^2}\right),
\end{equation}
where $\mathcal{G}$ is an $(N+M)\times (N+M)$ matrix, also called the
generalized Gaudin-determinant. It is the Jacobian associated to the
coupled set of nested BA equations and it is given 
by
\begin{equation*}
  \mathcal{G}=
  \begin{pmatrix}
    \GGG_{pp} & \GGG_{p\mu} \\
\GGG_{\mu p} & \GGG_{\mu\mu} 
  \end{pmatrix},
\end{equation*}
where the elements are
\begin{equation}
  \begin{split}
(\GGG_{pp})_{jk}&=\delta_{jk}\left(L-\sigma
\sum_{l}
  \frac{2c'}{(c')^2+(p_j-\mu_l)^2}
\right)+\\
&\hspace{2cm}
+\frac{1+\sigma}{2}\left(\delta_{jk}\sum_l
\frac{2c}{c^2+(p_j-p_l)^2}
-  \frac{2c}{c^2+(p_j-p_k)^2}\right)\\
(\GGG_{p\mu})_{jk}&=(\GGG_{\mu p})_{kj}=\sigma
\frac{2c'}{(c')^2+(p_j-\mu_k)^2}\\
(\GGG_{\mu\mu})_{jk}&=\delta_{jk}\left(
\sum_l \frac{2c'}{(c')^2+(p_l-\mu_j)^2}-
\sum_o \frac{2c}{c^2+(\mu_o-\mu_j)^2}
\right)+
\frac{2c}{c^2+(\mu_k-\mu_j)^2}
  \end{split}
\end{equation}
with $c'=c/2$.

\section{Conclusions and Outlook}

\label{sec:conclusions}

We investigated the form factors of local operators in the
multi-component Quantum Non-Linear Schr\"odinger equation. The main
 results of the present work are the following:
\begin{enumerate}
\item Establishing the analytic properties of the (infinite volume)
  form factors in the general $M$-component case; in particular the
   kinematical pole equation \eqref{infect-me2}. This can be
  regarded as a non-relativistic version of the well-known kinematical
  pole axiom from integrable relativistic QFT.
\item In the two-component case introducing the magnonic form factors and determining their
  analytic structure, in particular the kinematical pole equation
  \eqref{psybient}.
\item The solution of the equation  \eqref{psybient} in a number of simple
  cases, involving at most one magnonic rapidity per state. This was
  presented in section \ref{solutions}. 
\item Making a connection to the finite volume form factors,
  established by the Theorems \ref{fftcsa} and \ref{ugyanaz}, which
  state that the un-normalized form factors are exactly the same in
  finite and infinite volume. The normalized finite volume matrix
  elements are given by expressions \eqref{fanswer} and \eqref{full-answer}.
\end{enumerate}

The most interesting question seems to be whether or not new solutions
of the kinematical pole equations can be found. A possible direction
is to consider integral representations for the co-vector valued form
factors in the spirit of the so-called off-shell Bethe
Ansatz \cite{Karowski-LSZ}. However, it is not clear if such formulas
can be found or if they would be useful for
the calculation of correlation functions.

Another interesting question is to consider expectation values of
local (or non-local) operators in a finite volume. All the techniques
presented here apply in the case where there are no coinciding
particle rapidities; this lies at the heart of identifying the finite
volume and infinite volume matrix elements. On the other hand, it is
known from the one-component case that the mean values can always be
expressed with the properly regularized diagonal limits of the
infinite volume form factors. This leads to integral representations
for the mean values \cite{fftcsa2,LM-sajat}. A natural generalization
is to consider this problem in the multi-component case; research in
this direction is in progress.

It would be interesting to consider the singularity properties of form
factors in the framework of the Algebraic Bethe Ansatz. To our best
knowledge previous works only considered scalar products and norms of
Bethe states; our kinematical pole equation \eqref{psybient}
appears to be new. A very natural step would be to derive its ``finite
volume'' version from ABA. This could help in clarifying the relation
between the finite volume expectation values and the infinite volume
form factors.

In the end of Section \ref{2c} we noted that in the bosonic case the 
kinematical pole equation \eqref{psybient} does not contain enough
information to determine the form factors. It is an important open
question whether additional constraints can be found, which would make
the recursion equations constraining enough. Also, it is an
interesting question why does our formula \eqref{FFKM} work (at least
in the cases $N\le 5$) for the density operator of the down spin
particles. We did try other ways to generalize known formulas, and
only this one did reproduce the results from coordinate Bethe
Ansatz. This question is also left for further research.

As it was already mentioned in the introduction, the infinite volume
Quantum Inverse Scattering Method (QISM) yields a representation for the
field operator in terms of the Faddeev-Zamolodchikov operators
\cite{PhysRevD.21.1523}. From 
this result (the so-called quantum Rosales expansion) the form factors
can be simply read off using only the Faddeev-Zamolodchikov algebra
\cite{creamer-thacker-wilkinson-rossz}. The quantum Rosales expansion
is  established also in the multi-component case
\cite{Zhao-Pu-2comp-fermions-qism,multicomponent-bose-fermion-qism}
and this provides us an alternative way to obtain the form
factors. We checked in a number of simple cases (with low particle
number) that the results thus
obtained coincide with those presented in Section \ref{solutions}
\cite{nincskesz}. However, the 
QISM does not seem to lead to compact and manageable formulas for the generic
form factors with higher particle number. On the other hand, it would be interesting to consider
the two-point functions in this framework: it might be possible to
derive integral series for the finite temperature and finite density
correlations along the lines of
\cite{creamer-thacker-wilkinson-rossz}. 

In the present work we only considered the non-relativistic multi-component continuum
systems. However, it is expected that the ideas developed here
also apply to other model solvable by the nested Bethe Ansatz. In
particular it is expected that the annihilation pole equation \eqref{infect-me2}
(or its ``magnonic version'' \eqref{psybient} and its appropriate
generalizations to higher rank cases) hold in other $sl(N)$
related models, for example in the $SU(N)$ symmetric Heisenberg spin
chains. 

The identification of the finite and infinite volume form factors is
an important ingredient of the present work. It is very natural to ask:
is this result valid in integrable relativistic QFT? As it was
mentioned in the introduction, in the realm of integrable QFT the
infinite volume form factors are obtained using the so-called form
factor bootstrap program, which has been established for theories
with both diagonal and non-diagonal S-matrices. On the other hand,
less is known about the finite volume matrix elements.
In \cite{fftcsa1} it was
shown that a relation like our \eqref{fanswer}
(concerning the 1-component model) holds in the massive relativistic theories
with diagonal scattering. However, in the relativistic case the
equality is not exact: there are finite size effects (decaying
exponentially with the volume) due to virtual particle-antiparticle
pairs ``travelling around the world''. Concerning non-diagonal
scattering theories, for example the sine-Gordon model, it is
known that the finite size spectrum can be obtained (up to exponential
corrections) with 
essentially the same nested Bethe Ansatz construction, as applied in
this work \cite{Takacs-Feher1}.
The introduction
of the ``magnonic form factors'' is very natural also in the relativistic case,
and it is expected that a relation like \eqref{full-answer} holds as
well, once again up to the exponential corrections. As it was remarked
above, an interesting open question (both from the coordinate Bethe
Ansatz and the integrable QFT point of view) is the treatment of the matrix elements with
coinciding particle rapidities, or in other words the finite volume evaluation of
the ``disconnected pieces'' of the form factors. The resolution of
this question is important for the evaluation of finite temperature
correlations in massive integrable QFT \cite{fftcsa2,Essler:2009zz,D22,Doyon:finiteTreview}.

\vspace{1cm}
{\bf Acknowledgements} 

\bigskip

We are grateful to Jean-S\'ebastien Caux for useful discussions and
for bringing the paper \cite{impurity-spin-exp-koehl} to our
attention. Also, we acknowledge his contribution to finding our final formulas
for the form factors of the bilinear operators, expressed as
single determinants (see the footnote on page 8).

We are thankful to G\'abor Tak\'acs and
Jean-S\'ebastien Caux for useful comments about the manuscript.

B. P. was supported by the VENI Grant 016.119.023 of the NWO.

M. K. acknowledges funding from The Welch Foundation, Grant No. C-1739.

\addcontentsline{toc}{section}{References}
\bibliography{../../pozsi-general.bib}
\bibliographystyle{utphys}

\end{document}